\documentclass[11pt,reqno]{article}

\usepackage{amssymb,amscd,amsmath,amsthm,color,amsfonts,graphicx}
\usepackage{lineno}

\newtheorem{thm}{Theorem}[section]
\newtheorem{prop}[thm]{Proposition}
\newtheorem{lemma}[thm]{Lemma}
\newtheorem{cor}[thm]{Corollary}

\newtheorem{remark}[thm]{Remark}

\newtheorem{conj}[thm]{Conjecture}
 \newtheorem{exam}{Example}

\def\be{ \begin{linenomath} \begin{eqnarray}}
\def\ee{ \end{eqnarray} \end{linenomath}}
\def\bee{ \begin{linenomath} \begin{eqnarray*}}
\def\eee{\end{eqnarray*} \end{linenomath}}
\def\pmx{\begin{pmatrix}}
\def\emx{\end{pmatrix}}
\def\bsq{\begin{subequations}}
\def\esq{\end{subequations}}

\def\bst{\begin{subtheorem}}
\def\est{\end{subtheorem}}

\newcommand{\norm}[1]{ \| #1  \|}

\def\sym{{\rm sym}}
\def\riem{{\rm Riem}}
\def\relent{{\rm RelEnt}}
\def\geod{{\rm geod}}
\def\dob{{\rm Dob}}

\def\rt2{{\textstyle \frac{1}{\sqrt{2}}}}
\def\eps{\varepsilon}

\def\tr{\hbox{\rm Tr} \, }

\def\half{{\textstyle \frac{1}{2}}}
\def\nn{\nonumber}

\def\bra{\langle}
\def\ket{\rangle}
\def\kb{ \ket \bra }

\def\rt2{ \frac{1}{\sqrt{2}} }

\def\raw{\rightarrow}

\def\wh{\widehat}
\def\wtd{\widetilde}
\newcommand{\proj}[1]{ | #1 \kb  #1|}

\def\qed{\qquad {\bf QED}}
    
\def\RE{{\rm RelEnt}}
\def\riem{{\rm Riem}}
\def\dob{{\rm Dobrushin}}

\def\sym{{\rm sym}}

\def\geod{{\rm geod}}

\def\tr{{\rm Tr}}

\def\dtsig{{\mathbf \cdot \sigma}}
\def\bu{{\bf u}}
\def\bv{{\bf v}}
\def\bw{{\bf w}}

\def\bt{{\bf t}}
\def\by{{\bf y}}
\def\bx{{\bf x}}
\def\bz{{\bf z}}

\def\bN{\mathbb{N}}
\def\bM{\mathbb{M}}
\def\bH{\mathbb{H}}
\def\bP{\mathbb{P}}
\def\bC{\mathbf{C}}
\def\bR{\mathbf{R}}
\def\cK{\mathcal{K}}

\def\cG{\mathcal{G}}
\def\<{\langle}
\def\>{\rangle}
\def\Tr{\mathrm{Tr}\,}
\def\cD{\mathcal{D}}
\def\diag{\mathrm{diag}}

\def\WYD{\mathrm{WYD}}

\def\BKM{\mathrm{BKM}}
\def\supp{\mathrm{supp}\,}
\def\Bures{\mathrm{Bures}}
\def\bI{\mathbb{I}}
\def\cE{\mathcal{E}}
\def\cA{\mathcal{A}}
\def\kmin{{\kappa_{\min}}}
\def\kmax{{\kappa_{\max}}}
\def\WY{\mathrm{WY}}
\def\cP{\mathcal{P}}

\allowdisplaybreaks

\title{Contraction coefficients for noisy quantum channels\footnote{Dedicated to the memory of
  Professor  Uffe Haagerup}}
  
\author{Fumio Hiai \\ Tohoku University (Emeritus) \\
Hakusan 3-8-16-303, Abiko 270-1154, Japan \\
{\small hiai.fumio@gmail.com}\\ ~~ \\
\and Mary Beth Ruskai \\
Institute for Quantum Computing, University of Waterloo \\
Waterloo, Ontario, Canada \\
{\small ruskai@member.ams.org}}  

\date{\today}

\begin{document}

\maketitle

\begin{abstract}
Generalized relative entropy, monotone Riemannian metrics, geodesic distance,
and trace distance are all known to decrease under the action of quantum channels.
We give some new bounds on, and relationships between, the maximal contraction
for these quantities.


 
\end{abstract}

\bigskip


{\baselineskip=12pt
\tableofcontents
}


\section{Introduction} \label{sect:intro}

It is well-known that many quantities of interest in quantum information theory contract under
the action of completely positive and trace-preserving (CPT) maps, which represent quantum
channels, including the relative entropy $H(P,Q) \equiv\tr P (\log P - \log Q)$
of two positive definite operators with $\Tr P=\Tr Q$. When $\Phi$ is a quantum channel, we can
define 
\be \label{eta-relent}
\eta^\relent(\Phi) \equiv \sup\bigg\{ \frac{ H(\Phi(P),\Phi(Q)) }{ H(P,Q) } :
P, Q > 0, ~ P \neq Q, ~ \tr\,P = \tr\,Q \bigg\},
\ee
which describes the maximal contraction under $\Phi$. Another contraction coefficient can be
defined with respect to the trace distance as
$\eta^{\mathrm{Tr}}(\Phi)\equiv\sup\|\Phi(P-Q)\|_1/\|P-Q\|_1$, where the supremum is taken over
$P,Q$ as above. This can be regarded as the quantum version of the Dobrushin coefficient
of ergodicity \cite{Dob}. 

The concept of contraction coefficient was defined in the classical case \cite{CIRRSZ} and
generalized to the quantum setting in \cite{LR}. 
Similar definitions can be given to describe the contraction of many other quantities.
We consider here primarily   contraction with respect to quantum divergences (a special
case of quasi-entropies), monotone Riemannian metrics and geodesic distances arising from them.
There are many relations between the contraction coefficients of these quantities, which are our
main concern in this paper. We  also study the dependence of these contraction coefficients on
the particular operator convex functions used to define them. Both classical and quantum
contraction coefficients have important applications to the problem of mixing time bounds of
(quantum) Markov processes and in particular, (quantum) Markov chains, as demonstrated in, e.g.,
\cite{CIRRSZ,CKZ,KT2,TKRWV}.

We first  recall what is known in the classical setting. A classical channel from $\bC^d$
to $\bC^{d'}$ can be represented by a $d'\times d$ column-stochastic matrix $\Lambda$. The
trace-norm (or Dobrushin) contraction coefficient is
$\eta^{\mathrm{Tr}}(\Lambda)\equiv\sup\|\Lambda x\|_1/\|x\|_1$, where the supremum is taken
over non-zero $x\in\bR^d$ with $\sum_ix_i=0$. On the Riemannian manifold $\cP_d$ of probability
vectors $p=(p_1,\dots,p_d)$, $p_i>0$, $\sum_{i=1}^dp_i=1$, the so-called Fisher-Rao metric is
a unique classical monotone Riemannian metric, for which we have the Riemannian contraction
coefficient $\eta^\riem(\Lambda)$. When $g$ is a strictly convex function on $(0,\infty)$ with
$g(1)=0$, the classical $g$-divergence extending the classical relative entropy is defined as
$H_g(p,q)\equiv\sum_{i=1}^dg(p_i/q_i)q_i$ for $p,q\in\cP_d$, for which we have the contraction
coefficient $\eta_g^\relent(\Lambda)$. The following relations 
between these contraction coefficients were proved in \cite{CRS,CIRRSZ}:
\be \label{classical}
\eta_g^\relent(\Lambda)=\eta^\riem(\Lambda)\le\eta^{\mathrm{Tr}}(\Lambda)
\le\sqrt{\eta^\riem(\Lambda)}
\ee
whenever $g$ is operator convex on $(0,\infty)$.

In the quantum setting, the study  of monotone Riemannian metrics on the manifold
$\cD_d$ of $d\times d$ positive definite density matrices was begun by Morozova and \v{C}encov
\cite{MC}.  Petz  \cite{Pe3} then showed that there were infinitely many such metrics,
corresponding to positive operator monotone functions on $(0,\infty)$.      Following  
\cite{HKPR,LR} we use the set of operator convex functions $\kappa>0$ on $(0,\infty)$ with $\kappa(1)=1$
and $x\kappa(x)=\kappa(x^{-1})$ to parametrize symmetric monotone metrics on $\cD_d$, $d\in\bN$.
Such $\kappa$ functions correspond one-to-one, by $\kappa=1/f$, to operator monotone
functions $f>0$ on $(0,\infty)$ with $f(1)=1$ and $f(x)=xf(x^{-1})$ giving the same family
of such metrics as in \cite{Pe3}. Thus, for each
$\kappa$ function we can define the Riemannian contraction coefficient $\eta_\kappa^\riem(\Phi)$ of
a channel $\Phi$ from the $d\times d$ matrix algebra $\bM_d$ to $\bM_{d'}$, and
do so explicitly in \eqref{eta-Riem}
of Section 2.4. On the other hand, for each operator convex function $g$ on $(0,\infty)$ with
$g(1)=0$ and $g''(1)>0$ we have the quantum $g$-divergence $H_g(\rho,\gamma)$ for
$\rho,\gamma\in\cD_d$ and the corresponding contraction coefficient $\eta_g^\relent(\Phi)$, 
as defined in
\eqref{H_g} and \eqref{eta-RelEnt} of Section 2.3. As shown in
as well as \cite{LR}, and developed further here, the relation between $\eta_\kappa^\riem(\Phi)$ and
$\eta_g^\relent(\Phi)$ and their dependence on the $\kappa$ and $g$ functions in the quantum
setting are not as simple as  in the classical setting.

This paper is organized as follows. In Section 2 precise definitions of quantum
$g$-divergences and monotone metrics parametrized by the $\kappa$ functions are given, for
which we introduce the contraction coefficients $\eta_g^\relent(\Phi)$ and
$\eta_\kappa^\riem(\Phi)$. Section 3 provides familiar examples of $g$-divergences and monotone
metrics such as the BKM, the Wigner-Yanase, and the Bures metrics. In Section 4 a description of
$\eta_\kappa^\riem(\Phi)$ in terms of a certain eigenvalue problem developed in \cite{LR} is
recalled, which establishes the relation $\eta_\kappa^\riem(\Phi)\le\eta^{\mathrm{Tr}}(\Phi)$
when $\kappa(x)=x^{-1/2}$. 

The main results in  Section~5 are the general relations
\be \label{general}
\eta_\kappa^\geod(\Phi)=\eta_\kappa^\riem(\Phi)\le\eta_g^\relent(\Phi),
\qquad\eta^{\mathrm{Tr}}(\Phi)\le\sqrt{\eta_\kappa^\riem(\Phi)},
\ee
when $g$ is  related to $\kappa$ by $g(x)=(x-1)^2\kappa(x)$. Here
$\eta_\kappa^\geod(\Phi)$ is the contraction with respect to the geodesic distance induced by
the monotone metric for $\kappa$ as defined in \eqref{eta-geod}. A lemma slightly modified from \cite{HP1}
is given in Appendix A to prove the equality in \eqref{general}.  The first inequality in
\eqref{general} was given in \cite{LR}; the second is proved in 
Section~\ref{sect:gene-theorems} strengthening
results from  \cite{TKRWV,Ru}.  In Section~5.2, a partial ordering of contraction coefficients
is shown to hold when the domain or range of the channel is a commutative subalbegra, i.e.,
classical.  Some remarks on extensions to weak Schwarz maps are given in Section~5.3.

 In Section 6 we treat qubit channels using the 
Bloch sphere representation.  The section includes proofs of statements
announced in \cite{LR}, as well as  additional results.
 In particular when $\Phi_T$ is a unital qubit channel described 
by a $3 \times 3$ matrix $T$ as in \cite{KR}   we prove that 
\be \label{unital}
\eta_\kappa^{\riem}(\Phi_{T}) = \eta_\kappa^{\geod}(\Phi_{T}) = \eta_g^{\RE}(\Phi_{T})
= \| T \|_\infty ^2.
\ee
Next, for $\Phi$ in the simplest possible family of non-unital qubit channels 
(which collapse the Bloch sphere to a line), 
we estimate $\eta_\kappa^\riem(\Phi)$ for several particular cases of $\kappa$  as well as
$\eta_g^\relent(\Phi)$ for special $g$ corresponding to the extreme
$\kappa$ functions.   These estimates suffice to show that the equality conditions in \eqref{unital}
above do not extend to non-unital channels; that the contraction coefficients depend on
the functions $\kappa$ and $g$; and several natural conjectures are false. 
 Complete proofs of the results in Section 6, which are elementery  but somewhat lengthy,
 are given in Appendix~\ref{sect:qubitpf}

Finally in Section 7 we present further results on contraction coefficients for some special
examples of Section 3. A remarkable result here is that the equality
$\eta_\BKM^\riem(\Phi)=\eta_\BKM^\relent(\Phi)$ holds for every channel $\Phi$, where
$\eta_\BKM^\riem$ denotes the BKM metric contraction and $\eta_\BKM^\relent$ the contraction
with respect to the {\it symmetrized} relative entropy $H(P,Q)+H(Q,P)$. But the equality
between $\eta_\BKM^\riem(\Phi)$ and $\eta^\relent(\Phi)$ in \eqref{eta-relent} is left open.

\section{Notation and Definitions}

\subsection{Basic notation}\label{sect:notat}

For each $d\in\bN$ we write $\bM_d$, $\bH_d$, $\bP_d$, and $\overline{\bP}_d$ for the
sets of $d\times d$ complex, Hermitian, positive definite, and positive semi-definite
matrices, respectively. We also denote by $\cD_d$ the set of $d\times d$ positive definite
density matrices and $\overline\cD_d$ the set of all $d\times d$ density matrices, i.e.,
$\cD_d=\{\rho\in\bP_d:\Tr\rho=1\}$ and
$\overline\cD_d=\{\rho\in\overline\bP_d:\Tr\rho=1\}$, where $\Tr$ is the usual trace
functional on $\bM_d$. The trace-norm of $X\in\bM_d$ is $\|X\|_1\equiv \Tr|X|$. Recall that
$\bM_d$ identified with ${{\cal B}}({\bf C}^d)$ becomes a Hilbert space when
equipped with the Hilbert-Schmidt inner product
$$
\<X,Y\>\equiv \Tr X^*Y,\qquad X,Y\in\bM_d,
$$
together with the Hilbert-Schmidt norm $\|X\|_2\equiv (\Tr X^*X)^{1/2}$. A real subspace
$\bH_d$ of $\bM_d$ is identified with the Euclidean space of dimension $d^2$, and
$\cD_d$ is a smooth Riemannian manifold whose tangent space at any foot point is
identified with $\bH_d^0\equiv \{A\in\bH_d:\Tr A=0\}$. Functions $f(A)$ of matrices $A\in\bH_d$
are defined via the usual functional calculus.

We will use linear maps $\Phi:\bM_d\to\bM_{d'}$, and denote by $\widehat\Phi$ the adjoint of
$\Phi$ with respect to the Hilbert-Schmidt inner product, i.e., 
$\<\Phi(X),Y\>=\<X,\widehat\Phi(Y)\>$ for all $X\in\bM_d$ and $Y\in\bM_{d'}$.
As usual in quantum information, we call $\Phi$ a (quantum) {\it channel} if $\Phi$
is CPT (i.e., completely positive and trace-preserving) map. Most of the maps we consider
will be constructed from the {\it left} and {\it right multiplication operators}, respectively,
i.e., $L_AX\equiv AX$ and $R_BX\equiv XB$ for $A, B, X\in\bM_d$.
For each $A,B\in\bP_d$, $L_A$ and $R_B$ are commuting positive invertible operators on the
Hilbert space $\bM_d$ (however, they are not positive in the sense of mapping $\bP_d$
into $\overline\bP_d$). More generally, for functions $f: (0,\infty) \to {\bf R}$ we have
$L_{f(A)} = f(L_A)$ and $R_{f(B)} = f(B)$.


\subsection{Operator convex functions}

A real function $f$ on $(0,\infty)$ is said to be {\it operator monotone} (or operator
monotone increasing) if $A\ge B$ implies $f(A)\ge f(B)$ for every $A,B\in\bP_d$ with any
$d\in\bN$, and {\it operator monotone decreasing} if $-f$ is operator monotone. A real
function $g$ on $(0,\infty)$ is said to be {\it operator convex} if 
$$
g(\lambda A+(1-\lambda)B)\le\lambda g(A)+(1-\lambda)g(B)
$$
for all $A,B\in\bP_d$ with any $d\in\bN$ and all $\lambda\in(0,1)$, and {\it operator
concave} if $-g$ is operator convex.  The theory of operator monotone and operator convex
functions was initiated by L\"owner \cite{Lw} and Kraus \cite{Kr}, respectively. For
details, see, e.g., \cite[Section V.4]{Bh1}, also \cite{An1,Do,Hi}.

In this work the following classes of operator convex functions play a special role:
\begin{align*}
&\cG\equiv \{g:(0,\infty)\to\bR,\,\mbox{operator convex},\,g(1)=0,\,g''(1)>0\}, \\
&\cG_\sym\equiv \{g:(0,\infty)\to[0,\infty),\,\mbox{operator convex},
\,g(x)=xg(x^{-1})\ \mbox{for}\ x>0, \\
&\hskip8.8cm g(1)=g'(1)=0,\,g''(1)=2\}, \\
&\cK\equiv \{\kappa:(0,\infty)\to(0,\infty),\,\mbox{operator convex},
\,x\kappa(x)=\kappa(x^{-1})\ \mbox{for}\ x>0,\,\kappa(1)=1\}.
\end{align*}
By Proposition \ref{prop:kfg} below there is a one-to-one correspondence
$\kappa\in\cK\leftrightarrow g\in\cG_\sym$ determined by
\begin{equation}\label{kg-rel}
g(x)=(x-1)^2\kappa(x).
\end{equation}
It is easy to see that if $g\in\cG$ then $\widetilde g(x)\equiv xg(x^{-1})$ is also in $\cG$.
Indeed, if $g\in\cG$, then $g(x)/(x-1)=(g(x)-g(1))/(x-1)$ is operator monotone on
$(0,\infty)$ by Kraus' theorem, and hence
$$
{\widetilde g(x)-\widetilde g(1)\over x-1}
={xg(x^{-1})\over x-1}=-{g(x^{-1})\over x^{-1}-1}
$$
is also operator monotone so that $\widetilde g$ is operator convex on $(0,\infty)$.
Moreover, noting that $g''(1)=\widetilde g''(1)$, we define the {\it symmetrization} of $g$
by
\begin{equation}\label{g-sym}
g_\sym\equiv {g+\widetilde g\over g''(1)}\in\cG_\sym.
\end{equation}

\begin{prop}\label{prop:ope-conv}
{\rm(i)}\enspace
If $g:(0,\infty)\to\bR$ is an operator convex function, then there exist a unique constant
$c\ge0$ and a unique positive measure $\mu$ on $[0,\infty)$ with
$\int_{[0,\infty)}(1+s)^{-1}\,d\mu(s)<+\infty$ such that
\begin{equation}\label{int-exp1}
g(x)=g(1)+g'(1)(x-1)+c(x-1)^2
+\int_{[0,\infty)}{(x-1)^2\over x+s}\,d\mu(s),\quad x\in(0,\infty).
\end{equation}

{\rm(ii)}\enspace
If $\kappa:(0,\infty)\to\bR$ is an operator convex function and it satisfies the
normalization $\kappa(1)=1$ and the symmetry condition $x\kappa(x)=\kappa(x^{-1})$ for all
$x>0$, then $\kappa(x)>0$ for all $x>0$ (hence $\kappa\in\cK$) and there exists a unique
probability measure $m$ on $[0,1]$ such that
\begin{align}
\kappa(x)&=\int_{[0,1]}\frac{1+x}{(x+s)(1+sx)}\cdot\frac{(1+s)^2}{2}\,dm(s) \nonumber\\
&=\int_{[0,1]}\biggl(\frac{1}{x+s}+\frac{1}{sx+1}\biggr)
\frac{(1+s)}{2}\,dm(s),\qquad x\in(0,\infty). \label{int-exp2}
\end{align}
\end{prop}

The integral expression \eqref{int-exp1} was given in \cite{LR} and is a special case of
\cite[(5.2)]{FHR}. Since
\begin{equation}\label{g''(1)}
g''(1)=2\biggl(c+\int_{[0,\infty)}{1\over1+s}\,d\mu(s)\biggr),
\end{equation}
we note that $g''(1)>0$ if and only if $c+\mu([0,\infty))>0$, or equivalently, $g$ is not
a linear function. For the proof of \eqref{int-exp2}, see \cite[Appendix A.2]{HKPR}.
It is also obvious that $\kappa(x)>0$ for all $x>0$ whenever $\kappa$ is a convex function
with $\kappa(1)>0$ and $x\kappa(x)=\kappa(x^{-1})$ for all $x>0$.

By Proposition \ref{prop:ope-conv}, $\cK$ is a Bauer simplex (in a locally convex
topological vector space consisting of real functions on $(0,\infty)$ in the pointwise
convergence topology), whose extreme points are
\begin{equation}\label{ext-k}
\kappa_s(x)\equiv \frac{(1+s)^2}{2}\cdot\frac{1+x}{(x+s)(1+sx)}
=\frac{1+s}{2}\biggl(\frac{1}{x+s}+\frac{1}{1+sx}\biggr),
\qquad0\le s\le1.
\end{equation}
It is well-known (and immediately seen from the integral expression \eqref{int-exp2}) that
$\kappa_1(x)=2/(1+x)$ is the smallest element of $\cK$ and $\kappa_0(x)=(1+x)/2x$ is the
largest in $\cK$ so that
\begin{equation}\label{min-k-max}
{2\over1+x}\le\kappa(x)\le{1+x\over2x},\qquad\kappa\in\cK.
\end{equation}
In the sequel we use the more explicit notations $\kmin$ for $\kappa_1$ and $\kmax$ for
$\kappa_0$.

\begin{prop}\label{prop:kfg}
For a function $\kappa:(0,\infty)\to(0,\infty)$ consider the following conditions:
\begin{itemize}
\item[\rm(a)] $\kappa$ is operator convex,
\item[\rm(b)] $\kappa$ is operator monotone decreasing,
\item[\rm(c)] $g(x)\equiv (x-1)^2\kappa(x)$ is operator convex.
\end{itemize}
Then {\rm(a)} $\Leftarrow$ {\rm(b)} $\Leftrightarrow$ {\rm(c)}.
Moreover, if $x\kappa(x)=\kappa(x^{-1})$ for all $x>0$ or equivalently $g(x)=xg(x^{-1})$ for
all $x>0$, then the above conditions {\rm(a)--(c)} are all equivalent.
\end{prop}

\begin{proof}
For (b) $\Rightarrow$ (a), see \cite[Theorem 2.4]{HKPR}. As for (c), Kraus' theorem (see, e.g.,
\cite[Corollary 2.7.8]{Hi}) implies that $g$ is operator convex if and only if
$$
h(x)\equiv {g(x)-g(1)\over x-1}=(x-1)\kappa(x),\qquad x>0,
$$
is operator monotone. The latter is also equivalent to the condition that
$\kappa(x)=(h(x)-h(1))/(x-1)$ is operator monotone decreasing. Indeed, this is seen from
the facts that $h$ is operator monotone on $(0,\infty)$ if and only if it has an integral
expression
$$
h(x)=h(1)+\gamma(x-1)+\int_{[0,\infty)}{x-1\over x+s}\,d\mu(s),
\qquad x\in(0,\infty),
$$
where $\gamma\ge0$ and $\mu$ is a positive measure on $[0,\infty)$ with
$\int_{[0,\infty)}(1+s)^{-1}\,d\mu(s)<+\infty$ (see \cite[Theorem 1.9]{FHR}),
and that $\kappa$ is operator monotone decreasing on $(0,\infty)$ if and only if it has an
integral expression
$$
\kappa(x)=\gamma+\int_{[0,\infty)}{1\over x+s}\,d\mu(s),\qquad x\in(0,\infty),
$$
where $\gamma$ and $\mu$ are same as above (see \cite{Ha}, also \cite[Theorem 3.1]{AH}).

Next it is immediate to check that the conditions $x\kappa(x)=\kappa(x^{-1})$ and
$g(x)=xg(x^{-1})$ for all $x>0$ are equivalent. Under this symmetry condition, the implication
(a) $\Rightarrow$ (b) follows immediately from the integral expression \eqref{int-exp2}, as
shown in \cite[Theorem 2.4]{HKPR}. 
\end{proof}

\subsection{Relative entropy or $g$-divergence}

For every $g\in\cG$ and every $A,B\in\bP_d$ the (quantum) {\it $g$-divergence} of $A$
relative to $B$ is defined by
\be \label{H_g}
H_g(A,B)\equiv \bigl\<B^{1/2},g(L_AR_B^{-1})B^{1/2}\bigr\>,
\ee
which is a generalization of the relative entropy and is a special case of
{\it quasi-entropies} \cite{Ko,Pe1,Pe2}. The most important property of $H_g(A,B)$ is the
{\it monotonicity}
$$
H_g(\Phi(A),\Phi(B))\le H_g(A,B)
$$
for every $A,B\in\bP_d$ and every CPT map $\Phi:\bM_d\to\bM_{d'}$. This was first proved
by Petz \cite{Pe1,Pe2} under slightly more restricted situations,
 and the above extension is in \cite{LR}, see also \cite{TCR,HMPB}.

From the integral expression \eqref{int-exp1} it is known \cite[Theorem II.5]{LR} that
for every $\rho,\gamma\in\cD_d$,
\begin{equation}\label{H-int}
H_g(\rho,\gamma)=c\,\Tr(\rho-\gamma)^2\gamma^{-1}+\int_{[0,\infty)}
\Tr(\rho-\gamma){1\over L_\rho+sR_\gamma}(\rho-\gamma)\,d\mu(s).
\end{equation}
Since $(L_\rho+sR_\gamma)^{-1}$ is a positive invertible operator on $\bM_d$, the above
expression together with \eqref{g''(1)} implies that $H_g(\rho,\gamma)\ge0$ and that
$H_g(\rho,\gamma)=0$ if and only if $\rho=\gamma$.

In particular, when $g\in\cG_\sym$ and $\kappa\in\cK$ are given with \eqref{kg-rel}, we
note (see \cite[Theorem II.5]{LR}) that for every $A,B\in\bP_d$,
\begin{align*}
H_g(A,B)&=\bigl\<A-B,R_B^{-1}\kappa(L_AR_B^{-1})(A-B)\bigr\> \\
&=\bigl\<A-B,L_A^{-1}\kappa(R_BL_A^{-1})(A-B)\bigr\>=H_g(B,A).
\end{align*}
For every $g\in\cG$ with symmetrization $g_\sym$ in \eqref{g-sym} we have
\begin{equation}\label{H-g-sym}
H_{\widetilde g}(A,B)=H_g(B,A),\qquad H_{g_\sym}(A,B)={H_g(A,B)+H_g(B,A)\over g''(1)}.
\end{equation}

For each CPT map $\Phi:\bM_d\to\bM_{d'}$ and each $g\in\cG$ we introduce the
{\it contraction coefficient} of $\Phi$ with respect to the $g$-divergence $H_g$ by
\be \label{eta-RelEnt}
\eta_g^\relent(\Phi)\equiv \sup_{\rho,\gamma\in\cD_d,\,\rho\ne\gamma}
{H_g(\Phi(\rho),\Phi(\gamma))\over H_g(\rho,\gamma)}\quad(\le1).
\ee
From \eqref{H-g-sym} we easily see that
\begin{equation}\label{eta-g-sym}
\eta_{\widetilde g}^\relent(\Phi)=\eta_g^\relent(\Phi),\qquad
\eta_{g_\sym}^\relent(\Phi)\le\eta_g^\relent(\Phi).
\end{equation}

\subsection{Riemannian metrics and geodesic distance}

Given a function $\kappa\in\cK$ we define, for any $A\in\bP_d$, a linear map
$\Omega_A^\kappa:\bM_d\to\bM_d$ by
$$
\Omega_A^\kappa(X)\equiv R_A^{-1}\kappa\bigl(L_AR_A^{-1}\bigr)X
=L_A^{-1}\kappa\bigl(R_AL_A^{-1}\bigr)X,
\qquad X\in\bM_d,
$$
where the equality of the two expressions follows from $x\kappa(x)=\kappa(x^{-1})$. The
positivity condition in the sense that
$\Omega_A^\kappa(\bP_d)\subset\overline\bP_d$ (equivalent to complete positivity) for the
map $\Omega_A^\kappa$ has thoroughly been investigated in \cite{HKPR} with a lot of sample
discussions.

Associated with $\kappa\in\cK$ a Riemannian metric $M^\kappa$ on the Riemannian manifold
$\cD_d$ is defined by
\begin{equation}\label{Riem}
M_\rho^\kappa(A,B)\equiv \bigl\<A,\Omega_\rho^\kappa(B)\bigr\>,
\qquad A,B\in\bH_d^0,\ \rho\in\cD_d.
\end{equation}
This family of Riemannian metrics on $\cD_d$ ($d\in\bN$) induced by $\kappa\in\cK$ is
called {\it monotone metrics} since the class was characterized by Petz \cite{Pe3} with
the monotonicity property:
\begin{equation}\label{Riem-mono}
M_{\Phi(\rho)}(\Phi(A),\Phi(A))\le M_\rho(A,A),
\qquad A\in\bH_d^0,\ \rho\in\cD_d,
\end{equation}
for every CPT map $\Phi:\bM_d\to\bM_{d'}$ with any $d,d'$. Here, note that although
$\Phi(\rho)$ in $\bM_{d'}$ is not necessarily positive definite, the left-hand side of
\eqref{Riem-mono} is well defined by regarding $\Phi(\rho)$ and $\Phi(A)$ as matrices in
$\Pi\bM_{d'}\Pi\cong\bM_k$ where $\Pi\equiv \supp\Phi(I_d)$, the support projection, and
$k\equiv \dim\Pi$. More recent results on monotone metrics are found in \cite{HP1,HP2}.
Also, note that if $\rho$ and $A$ commute, then $\<A,\Omega_\rho^\kappa(A)\>=\Tr\rho^{-1}A^2$
independently of the choice of $\kappa\in\cK$. This fact is essentially same as the classical
result that there is only one monotone metric in the classical setting, known as the
Fisher-Rao metric.

For each $\kappa\in\cK$ the {\it contraction coefficient} of a CPT map $\Phi$ with respect
to the monotone metric $M^\kappa$ induced by $\kappa$ is defined by
\be \label{eta-Riem}
\eta_\kappa^\riem(\Phi)\equiv \sup_{\rho\in\cD_d}\,\sup_{A\in\bH_d^0,\,A\ne0}
{\bigl\<\Phi(A),\Omega_{\Phi(\rho)}^\kappa(\Phi(A))\bigr\>\over\<A,\Omega_\rho^\kappa(A)\>}
\quad(\le1).
\ee

We write $\eta_{\max}^\riem(\Phi) $ and $\eta_{\min}^\riem(\Phi)$ for the contraction
coefficients associated with the largest $\kmax(x)=(1+x)/2x$ and smallest $\kmin(x)=2/(1+x)$
functions in $\cK$, while these need not be the largest and smallest contraction coefficients
for a given $\Phi$. (Indeed, Proposition \ref{prop:semi-classical} below suggests that this
is not true in general.)


We write $D_\kappa(\rho,\gamma)$ for the {\it geodesic distance} with respect to $M^\kappa$;
namely,
\begin{equation}\label{geod}
D_\kappa(\rho,\gamma)
\equiv \inf_\xi\int_0^1\sqrt{\bigl\<\xi'(t),\Omega_{\xi(t)}^\kappa(\xi'(t))\bigr\>}\,dt,
\end{equation}
where the infimum is taken over all (piecewise) smooth curves joining $\rho,\gamma$ in $\cD_d$.
The monotonicity of $M^\kappa$ obviously implies that
$$
D_\kappa(\Phi(\rho),\Phi(\gamma))\le D_\kappa(\rho,\gamma)
$$
for every CPT map $\Phi:\bM_d\to\bM_{d'}$. The {\it contraction coefficient} of $\Phi$
with respect to the geodesic distance $D_\kappa$ is then defined by
\be \label{eta-geod}
\eta_\kappa^\geod(\Phi)\equiv \Biggl[\sup_{\rho,\gamma\in\cD_d,\,\rho\ne\gamma}
{D_\kappa(\Phi(\rho),\Phi(\gamma))\over D_\kappa(\rho,\gamma)}\Biggr]^2
\quad(\le1).
\ee
We will prove the equality $\eta_\kappa^\riem(\Phi)=\eta_\kappa^\geod(\Phi)$ in Section
\ref{sect:gene-theorems} while the inequality
$\eta_\kappa^\riem(\Phi)\ge\eta_\kappa^\geod(\Phi)$ was shown in \cite[Theorem IV.2]{LR}.

\begin{remark}\label{remark:geod} \rm
The space $\bP_d$ ($\supset\cD_d$) is a smooth Riemannian manifold with the tangent space
$\bH_d$ ($\supset\bH_d^0$). For each $\kappa\in\cK$ a Riemannian metric $M^\kappa$ on $\bP_d$
is defined by the same expression as \eqref{Riem} for $A,B\in\bH_d$ and $\rho\in\bP_d$. For
any CPT map $\Phi$, if $A=\rho\in\cD_d$ then we have
$$
{\bigl\<\Phi(A),\Omega_{\Phi(\rho)}^\kappa(\Phi(A))\bigr\>\over\<A,\Omega_\rho^\kappa(A)\>}
={\Tr\Phi(\rho)\over\Tr\rho}=1,
$$
which says that the contraction coefficient of $\Phi$ with respect to $M^\kappa$ is meaningless
if it is defined with $\bH_d$ instead of $\bH_d^0$.

For any pair $\rho,\gamma\in\cD_d$, in addition to  \eqref{geod}
one can consider the geodesic distance $\widetilde D_\kappa(\rho,\gamma)$ with the same
expression but taken over all smooth curves joining $\rho,\gamma$ in $\bP_d$ without
confining them to
$\cD_d$. The difference between the manifolds $\cD_d$ and $\bP_d$ implies that
$\widetilde D_\kappa(\rho,\gamma) \leq D_\kappa(\rho,\gamma)$.
Moreover,  for every $\kappa\in\cK$ this inequality will be strict when $\rho\gamma=\gamma\rho$
and $\rho\ne\gamma$.  In this case
$D_\kappa(\rho,\gamma)$ and $\widetilde D_\kappa(\rho,\gamma)$ are independent of the choice
of $\kappa\in\cK$ so that the formulas \eqref{geod-WY} and \eqref{geod2-WY} of the next section
hold for any $\kappa\in\cK$ as does the inequality \eqref{strict.geod}.
Indeed, one can apply the monotonicity of $M^\kappa$ to the trace-preserving
conditional expectation onto a commutative subalgebra $\cA$ containing $\rho,\sigma$ to see
that curves $\xi$ can be confined, in the definition \eqref{geod}, to those inside $\cA$.
(See also the proof of Lemma \ref{lemma:trace-Riem}.)

\end{remark}

%

\section{Examples}

Basic examples of $g$-divergences and monotone metrics are in order here. Further discussions
and results for these cases will be later given in Section 7.

\begin{exam} {\rm (Relative entropy and BKM metric)\enspace The function $g(x)=x\log x \in\cG$
gives the (usual) {\it relative entropy}, i.e.,
$$
H_{x\log x}(\rho,\gamma)=H(\rho,\gamma)\equiv \Tr \rho(\log \rho-\log \gamma).
$$
Moreover, $\wtd{g}(x)=-\log x$ and $H_{-\log x}(\rho,\gamma)=H(\gamma,\rho)$. The symmetrization
$$
g_\BKM(x)\equiv x \log x - \log x\in\cG_\sym
$$
corresponds to the function
$$
\kappa_\BKM(x)\equiv {\log x\over x-1}\in\cK,
$$
which gives
\begin{equation}\label{Omega:BKM}
\Omega_\rho^\BKM(X)={\log L_\rho-\log R_\rho\over L_\rho-R_\rho}(X)
=\int_0^\infty{1\over\rho+tI}\,X\,{1\over\rho+tI}\,dt.
\end{equation}
The corresponding monotone metric is the so-called {\it Bogolieubov} (or {\it Kubo-Mori})
{\it metric}. We write $\eta_\BKM^\riem(\Phi)$ and $\eta_\BKM^\relent(\Phi)$ for the contraction
coefficients associated with $\kappa_\BKM$ and $g_\BKM$. Rather surprisingly, it will be shown
in Theorem~\ref{thm:BKM} that $\eta_\BKM^\riem(\Phi)=\eta_\BKM^\relent(\Phi)$ holds for every
CPT map $\Phi$. However, we know from Theorem~\ref{thm:relent.neq.riem} that this property does
not hold in general.
} \end{exam}

\begin{exam} \label{ex:max} {\rm (Maximal metric)\enspace  The function $g(x) = (x-1)^2 \in \cG$
yields the quadratic relative entropy  
\begin{equation}\label{H-quad}
H_{(x-1)^2}(\rho,\gamma) = \Tr (\rho - \gamma)^2\gamma^{-1} = \Tr \rho^2 \gamma^{-1} - 1.
\end{equation}
The function $g_{\max}(x)\equiv(x-1)^2(1+x)/2x$ in $\cG_\sym$ is the symmetrization of $g$,
which corresponds to the largest function $\kmax$ in $\cK$. The function $\kmax$ defines the
largest monotone metric with
$$
\Omega_\rho^{\max}(X)={L_\rho^{-1}+R_\rho^{-1}\over2}(X)
={\rho^{-1}X+X\rho^{-1}\over2}.
$$
Note that for the  choice $ X = \rho - \gamma $,
\be  \label{quadspec}
\bigl\<\rho-\sigma, \Omega^{\max}_{\rho}(\rho-\sigma)\bigr\> = H_{(x-1)^2}(\gamma,\rho)
\ee
} \end{exam}

\begin{exam} \label{ex:pow} {\rm (Central power metric)\enspace
The function $x^{-1/2}\in\cK$ gives
$$
\Omega_\rho^{x^{-1/2}}(X)=\rho^{-1/2}X\rho^{-1/2},
$$
which may be considered as the center of $\cK$ from some aspects. For instance, it is known
\cite[Theorem 3.5]{HKPR} that $x^{-1/2}$ is the only function $\kappa\in\cK$ such that both
$\Omega_\rho^\kappa$ and $(\Omega_\rho^\kappa)^{-1}$ are CP for all $\rho\in\cD_d$. The function
in $\cG_\sym$ corresponding to $x^{-1/2}$ is $(x-1)^2x^{-1/2}$ and the corresponding divergence
is
$$
H_{(x-1)^2x^{-1/2}}(\rho,\gamma)=\Tr(\rho-\gamma)\rho^{-1/2}(\rho-\gamma)\gamma^{-1/2}.
$$
} \end{exam}

\begin{exam} \label{ex:WYD} {\rm (Wigner-Yanase-Dyson metric)\enspace
For any $t\in(0,1)\cup(1,2]$ define
\begin{equation}\label{g^t}
g^{(t)}(x)\equiv{x-x^t\over t(1-t)}\in\cG,
\end{equation}
whose symmetrized function corresponds to
\begin{equation}\label{k-WYD}
\kappa_t^\WYD(x)\equiv{1\over t(1-t)}\cdot{(1-x^t)(1-x^{1-t})\over(1-x)^2}\in\cK.
\end{equation}
Note that the functions $\kappa_t^\WYD$ extend to the parameter $t\in[-1,2]$ with
$\kappa_t^\WYD=\kappa_\BKM$ for $t=0,1$ by taking the limit as $t\to0,1$ and with
$\kappa_t^\WYD=\kappa_{1-t}^\WYD$ as symmetric around $t=1/2$. For the particular case $t=1/2$
the monotone metric associated with $\kappa_\WY\equiv\kappa_{1/2}^\WYD$ is called the
{\it Wigner-Yanase metric} with
$$
\Omega_\rho^\WY={4\over\bigl(\sqrt{L_\rho}+\sqrt{R_\rho}\bigr)^2},
$$
and the divergence for $g^{(1/2)}(x)=4(x-\sqrt x)$ is
$$
H_{4(x-\sqrt x)}(\rho,\gamma)=4(1-\Tr\rho^{1/2}\gamma^{1/2}).
$$

It seems that Hasegawa \cite{Has} was the first to realize the WYD metric as well as the WYD
divergences could be extended to the full parameter range $[-1,2]$. See also \cite{JR} where
equality conditions were given for the convexity of $g$-divergences for the WYD functions.

It is known \cite[Theorem 5.4]{GI} that for every $\rho,\gamma\in\cD_d$ the geodesic distance
$D_\WY(\rho,\gamma)$ with respect to the metric for $\kappa_\WY$ is given as
\begin{equation}\label{geod-WY}
D_\WY(\rho,\gamma)=\arccos\Tr\rho^{1/2}\gamma^{1/2}.
\end{equation}
On the other hand, the geodesic distance $\widetilde D_\WY(\rho,\gamma)$ taken over curves in
$\bP_d$ (see Remark \ref{remark:geod}) is included in \cite[Theorem 2.1]{HP0} and we have
\begin{equation}\label{geod2-WY}
\widetilde D_\WY(\rho,\gamma)=\big\|\rho^{1/2}-\gamma^{1/2}\big\|_2
=\sqrt{2-2\Tr\rho^{1/2}\gamma^{1/2}}.
\end{equation}
Since $\sqrt{2-2t}<\arccos t$ for $0\le t<1$, we see that
\be    \label{strict.geod}
\widetilde D_\WY(\rho,\gamma)<D_\WY(\rho,\gamma)
\ee
 unless $\rho=\gamma$.
} \end{exam}

\begin{exam} \label{ex:min} {\rm (Minimal or Bures metric)\enspace
The smallest function $\kmin(x)=2/(1+x)$ in $\cK$ defines the smallest monotone metric with
\begin{equation}\label{Omega:min}
\Omega_\rho^{\min}(X)={2\over L_\rho+R_\rho}(X)
=2\int_0^\infty e^{-t\rho}Xe^{-t\rho}\,dt,
\end{equation}
which is often called the {\it SLD metric} (symmetric logarithmic derivative). This is
considered as the infinitesimal form of the {\it Bures distance} introduced in \cite{Bur} and
was intensively studied by Uhlmann, e.g., \cite{Uh2}, so it is also called the {\it Bures}
or {\it Bures-Uhlmann metric}. The corresponding function in $\cG_\sym$ is
$g_{\min}(x)\equiv2(x-1)^2/(x+1)$ and the corresponding divergence is
$$
H_{\min}(\rho,\gamma)
=\biggl\<\rho-\gamma,{2\over L_\rho+R_\gamma}(\rho-\gamma)\biggr\>.
$$

Recall that the Bures distance \cite{Bur} between $\rho,\gamma\in\overline{\cD}_d$ is
\begin{equation}\label{dist-Bures}
d_\Bures(\rho,\gamma)\equiv\sqrt{2-2F(\rho,\gamma)},
\end{equation}
where $F(\rho,\gamma)\equiv\Tr(\rho^{1/2}\gamma\rho^{1/2})^{1/2}$ is the {\it fidelity} of
$\rho,\gamma$. It is known (see \cite{Uh3} and \cite[(9.32)]{BZ}) that the geodesic distance
between $\rho,\gamma\in\cD_d$ with respect to the Bures metric is given as
\begin{equation}\label{geod-Bures}
D_{\min}(\rho,\gamma)=\arccos F(\rho,\gamma).
\end{equation}
Since $\Tr\rho^{1/2}\gamma^{1/2}<F(\rho,\gamma)$ unless $\rho\gamma=\gamma\rho$
(see \cite[Corollary 3.4]{FS} and \cite[Theorem 2.1]{Hi0}), by comparing \eqref{geod-WY} and
\eqref{geod-Bures} one can see that $D_\WY(\rho,\gamma)>D_{\min}(\rho,\gamma)$ whenever
$\rho\gamma\ne\gamma\rho$.

} \end{exam}

\section{Trace Distance and Eigenvalue Formulation}

\subsection{Trace distance}

The most widely used distance for density matrices is the trace-norm distance
$\|\rho-\gamma\|_1=\Tr|\rho-\gamma|$. The next trace-norm monotonicity property is a slight
extension of \cite[Theorem 1]{Ru}. We give a proof for completeness.

\begin{prop}
Let $\Phi:\bM_d\to\bM_{d'}$ be a positive (not necessarily CP) trace-preserving map. Then
\begin{equation}\label{trace-cont}
\|\Phi(X)\|_1\le\|X\|_1,\qquad X\in\bM_d.
\end{equation}
\end{prop}

\begin{proof}
Note that the adjoint map $\widehat\Phi:\bM_{d'}\to\bM_d$ is positive and unital. Due to
the Russo-Dye theorem (see, e.g., \cite[Theorem 2.3.7]{Bh2}) we have
$$
\|\widehat\Phi(Z)\|_\infty\le\|Z\|_\infty,\qquad Z\in\bM_{d'},
$$
where $\|\cdot\|_\infty$ denotes the operator norm. Since $\|\cdot\|_1$ is the dual norm
of $\|\cdot\|_\infty$, we have the asserted inequality as follows:
\begin{align*}
\|\Phi(X)\|_1
&=\sup\bigl\{|\<X,\widehat\Phi(Z)\>|:Z\in\bM_{d'},\,\|Z\|_\infty\le1\bigr\} \\
&\le\sup\bigl\{|\<X,Y\>|:Y\in\bM_d,\,\|Y\|_\infty\le1\bigr\}=\|X\|_1
\end{align*}
for every $X\in\bM_d$.
\end{proof}

The {\it contraction coefficient} of a positive trace-preserving map $\Phi$ with respect to
the trace-norm distance is defined by
\be \label{eta-Tr}
\eta^\dob(\Phi)=\eta^{\mathrm{Tr}}(\Phi)\equiv \sup_{\rho,\gamma\in\cD_d,\,\rho\ne\gamma}
{\|\Phi(\rho)-\Phi(\gamma)\|_1\over\|\rho-\gamma\|_1}
=\sup_{A\in\bH_d^0,\,A\ne0}{\|\Phi(A)\|_1\over\|A\|_1}\quad(\le1),
\ee
which is the quantum generalization of the classical {\it Dobrushin coefficient of
ergodicity}. It was shown in \cite[Theorem 2]{Ru} that
\begin{equation}\label{Ruskai-form}
\eta^{\mathrm{Tr}}(\Phi)={1\over2}\sup\bigl\{\|\Phi(E-F)\|_1:
E,F\in\bM_d,\,\mbox{rank $1$ projections},\,E\perp F\bigr\}.
\end{equation}

 Next we extend the notion of {\it scrambling} column-stochastic
matrices in \cite{CIRRSZ} to the matrix algebra setting.  
\begin{prop}\label{prop:scram}
Let $\Phi:\bM_d\to\bM_{d'}$ be a positive trace-preserving map. Then the following
conditions are equivalent:
\begin{itemize}
\item[\rm(i)] $\eta^{\mathrm{Tr}}(\Phi)<1$;
\item[\rm(ii)] for every rank $1$ projections $E,F\in\bM_d$ with $E\perp F$,
$\Tr\Phi(E)\Phi(F)>0$;
\item[\rm(iii)] for every non-zero $A,B\in\overline\bP_d$, $\Tr\Phi(A)\Phi(B)>0$.
\end{itemize}
\end{prop}

\begin{proof}
We denote by $\Sigma$ the set of pairs $(E,F)$ of rank $1$ projections in $\bM_d$ with
$E\perp F$. Since $\Sigma$ is compact in $\bM_d\times\bM_d$, there is an
$(E_0,F_0)\in\Sigma$ such that
$\|\Phi(E_0)-\Phi(F_0)\|_1=\sup_{(E,F)\in\Sigma}\|\Phi(E)-\Phi(F)\|_1$. Here, note that,
for any $\rho,\gamma\in\overline\cD_d$, $\|\rho-\gamma\|_1=2$ if and only if the support
projections of $\rho,\gamma$ are orthogonal, if and only if $\Tr\rho\gamma=0$. Hence
(i) $\Leftrightarrow$ (ii) is immediate from \eqref{Ruskai-form}.
That (iii) $\Rightarrow$ (ii) is trivial. To prove that (i) $\Rightarrow$ (iii), assume
(i) and let $A,B\in\overline\bP_d$ be non-zero. Set $\rho\equiv A/\Tr A$ and
$\gamma\equiv B/\Tr B$; then $\rho,\gamma\in\overline\cD_d$ and so
$\Phi(\rho),\Phi(\gamma)\in\overline\cD_{d'}$. If $\rho=\gamma$, then
$\Tr\Phi(A)\Phi(B)>0$ is clear. If $\rho\ne\gamma$, then assumption (i) implies that
$\|\Phi(\rho)-\Phi(\gamma)\|_1<\|\rho-\gamma\|_1\le2$. Hence we have
$\Tr\Phi(\rho)\Phi(\gamma)>0$, i.e., $\Tr\Phi(A)\Phi(B)>0$.
\end{proof}

\subsection{Eigenvalue formulation}

We here summarize for the convenience of the reader, an observation in \cite{LR} to link
the Riemannian metric coefficient with an eigenvalue problem, which is of some interest in
its own right. Moreover, it allows us to connect  $\eta^\tr(\Phi) $ with $\eta_g^\riem(\Phi)$
in some very special cases, as in Theorem~\ref{thm:Riem-trace} below.

Let $\rho\in\cD_d$ and $\Phi:\bM_d\to\bM_{d'}$ be a CPT map. Define a linear map
$\Psi:\bM_d\to\bM_{d'}$ by
$$
\Psi=\Psi_{\Phi,\rho}^\kappa
\equiv (\Omega_{\Phi(\rho)}^\kappa)^{1/2}\Phi(\Omega_\rho^\kappa)^{-1/2}.
$$
Since $\Psi$ is a contraction and
$\widehat\Psi\Psi(\Omega_\rho^\kappa)^{-1/2}(I_d)=(\Omega_\rho^\kappa)^{-1/2}(I_d)$,
note that $(\Omega_\rho^\kappa)^{-1/2}(I_d)$ is an eigenvector of $\widehat\Psi\Psi$
corresponding to the largest eigenvalue $1$. Let $\lambda_2^\kappa(\Phi,\rho)$ denote the
second largest eigenvalue (with multiplicities counted) of $\widehat\Psi\Psi$. Then
$\lambda_2^\kappa(\Phi,\rho)$ is represented as
\begin{equation}\label{lambda_2}
\lambda_2^\kappa(\Phi,\rho)
=\sup_{X\in\bM_d^0,\,X\ne0}{\bigl\<\Phi(X),\Omega_{\Phi(\rho)}^\kappa(\Phi(X))\bigr\>
\over\bigl\<X,\Omega_\rho^\kappa(X)\bigr\>},
\end{equation}
where $\bM_d^0\equiv \{X\in\bM_d:\Tr X=0\}$. From the fact that $\Omega_\rho^\kappa(A)\in\bH_d$
and $\<A,\Omega_\rho^\kappa(B)\>=\<B,\Omega_\rho^\kappa(A)\>$ for all
$A,B\in\bH_d$, one can easily see that the right-hand side of \eqref{lambda_2} coincides
with $\eta_\kappa^\riem(\Phi)$. Therefore, we have

\begin{thm}\label{thm:Riem-lambda}{\bf\cite[Theorem IV.4]{LR}}\enspace
For every $\kappa\in\cK$ and every CPT map $\Phi$,
$$
\eta_\kappa^\riem(\Phi)=\sup_{\rho\in\cD_d}\lambda_2^\kappa(\Phi,\rho).
$$
\end{thm}

Now assume that an eigenvector of $\widehat\Psi\Psi$ corresponding to
$\lambda_2\equiv \lambda_2^\kappa(\Phi,\rho)$ is given by $(\Omega_\rho^\kappa)^{1/2}(X)$,
orthogonal to $(\Omega_\rho^\kappa)^{-1/2}(I_d)$. Then $\Tr X=0$ and
$\widehat\Psi\Psi(\Omega_\rho^\kappa)^{1/2}(X)=\lambda_2(\Omega_\rho^\kappa)^{1/2}(X)$,
which is equivalently written as
\begin{equation}\label{eigen-prob}
(\Omega_\rho^\kappa)^{-1}\widehat\Phi\Omega_{\Phi(\rho)}^\kappa(\Phi(X))=\lambda_2X.
\end{equation}
Thus, finding $\lambda_2^\kappa(\Phi,\rho)$ is equivalent to solving the eigenvalue
problem \eqref{eigen-prob} under the constraint $\Tr X=0$. In connection with
\eqref{eigen-prob} we define a linear map $\Upsilon:\bM_{d'}\to\bM_d$ by
$$
\Upsilon=\Upsilon_{\Phi,\rho}^\kappa
\equiv (\Omega_\rho^\kappa)^{-1}\widehat\Phi\Omega_{\Phi(\rho)}^\kappa.
$$
Since $\widehat\Upsilon(I_d)=I_{d'}$, $\Upsilon$ is trace-preserving. Here, assume that
$\Upsilon$ is positive in the sense that $\Upsilon(\bP_{d'})\in\overline\bP_d$.
Then, thanks to \eqref{trace-cont}, $\|\Upsilon(Z)\|_1\le\|Z\|_1$ for all $Z\in\bM_{d'}$.
Therefore, if $X\in\bM_d^0$ is a solution of the eigenvalue equation \eqref{eigen-prob},
then we have $\lambda_2\|X\|_1=\|\Upsilon(\Phi(X))\|_1\le\|\Phi(X)\|_1$ so that
$\lambda_2\le\|\Phi(X)\|_1/\|X\|_1$. Since all linear maps involving in
\eqref{eigen-prob} are self-adjoint, note that a solution $X$ of \eqref{eigen-prob}
with $\Tr X=0$ can always be taken in $\bH_d^0$.

From the above argument, if $\Upsilon$ is positive for every $\rho\in\cD_d$, then we
would have $\eta_\kappa^\riem(\Phi)\le\eta^{\mathrm{Tr}}(\Phi)$. This situation indeed
occurs, in particular, when both $\Omega_\rho^\kappa$ and $(\Omega_\rho^\kappa)^{-1}$ are
positive (equivalently, CP) for every $\rho\in\cD_d$ and every $d\in\bN$. But it is
known \cite[Proposition 3.5]{HKPR} that this latter condition holds only when
$\kappa(x)=x^{-1/2}$. So we have

\begin{thm}\label{thm:Riem-trace}{\bf\cite[Theorem 14]{TKRWV}}\enspace
For every CPT map $\Phi$,
$$
\eta_{x^{-1/2}}^\riem(\Phi)\le\eta^{\mathrm{Tr}}(\Phi).
$$
\end{thm}

\section{General Contraction Results}

\subsection{Results for arbitrary channels} \label{sect:gene-theorems}

In this subsection we present a few general relations between the contraction coefficients
defined in Section 2. The next theorem says the general equality between the
Riemannian metric contraction coefficient and the geodesic contraction coefficient.
The proof is based on a limit formula in \cite{HP0} for the geodesic distance,
whose proof is presented  in Appendix \ref{sect:appendA} for completeness.

\begin{thm}\label{thm:Riem-geod}
For every $\kappa\in\cK$ and every CPT map $\Phi:\bM_d\to\bM_{d'}$,
$$
\eta_\kappa^\riem(\Phi)=\eta_\kappa^\geod(\Phi).
$$
\end{thm}

\begin{proof}
The inequality $\eta_\kappa^\geod(\Phi)\le\eta_\kappa^\riem(\Phi)$ was shown in
\cite[Theorem IV.2]{LR}. To prove the reverse inequality, we use Lemma
\ref{lemma:limit-form} in the appendix. For every $\rho\in\cD_d$ and every $A\in\bH_d^0$
with $A\ne0$, by the lemma we have
$$
{\bigl\<\Phi(A),\Omega_{\Phi(\rho)}^\kappa(\Phi(A))\bigr\>\over
\<A,\Omega_\rho^\kappa(A)\>}
=\lim_{\eps\searrow0}\biggl[{D_\kappa(\Phi(\rho),\Phi(\rho+\eps A))\over
D_\kappa(\rho,\rho+\eps A)}\biggr]^2\le\eta_\kappa^\geod(\Phi),
$$
which implies that $\eta_\kappa^\riem(\Phi)\le\eta_\kappa^\geod(\Phi)$.
\end{proof}

For completeness we state the following theorem.  The first inequality was proved in \cite{LR}, and the
rest is a straightforward consequence of \eqref{H-g-sym} as discussed earlier. 
 \begin{thm}\label{thm:Riem-RelEnt}{\bf\cite[Theorem IV.2]{LR}}\enspace
For every $g\in\cG$ let $g_\sym$ be the symmetrization of $g$ as in \eqref{g-sym} and
$\kappa(x)\equiv g_\sym(x)/(x-1)^2\in\cK$ as in \eqref{kg-rel}. Then, for every CPT map
$\Phi:\bM_d\to\bM_{d'}$,
$$
\eta_\kappa^\riem(\Phi)\le\eta_{g_\sym}^\relent(\Phi)\le\eta_g^\relent(\Phi)
= \eta_{\wtd{g}}^\relent(\Phi).
$$
\end{thm}

In the next theorem we give a general inequality between the contraction coefficients
for Riemannian metrics and the trace-norm,  generalizng
 \cite[Theorem 3]{Ru}  and  \cite[Theorem 13]{TKRWV}.  
\begin{thm}\label{thm:trace-Riem}
For every $\kappa\in\cK$ and every CPT map $\Phi:\bM_d\to\bM_{d'}$,
$$
\eta^{\mathrm{Tr}}(\Phi)\le\sqrt{\eta_\kappa^\riem(\Phi)}.
$$
\end{thm}
To prove this result, we first give a lemma generalizing \cite[Lemma 5]{TKRWV}.
\begin{lemma}\label{lemma:trace-Riem}
For every $\kappa\in\cK$ and every $\rho\in\cD_d$,
$$
\|A\|_1^2\le\<A,\Omega_\rho^\kappa(A)\>,\qquad A\in\bH_d^0.
$$
\end{lemma}
\begin{proof}
Let $\cE$ be the trace-preserving conditional expectation from $\bM_d$ onto the subalgebra
generated by $A$. The monotonicity of $M^\kappa$ implies
that
$$
\<A,\Omega_\rho^\kappa(A)\>\ge\bigl\<A,\Omega_{\cE(\rho)}^\kappa(A)\bigr\>
=\Tr\cE(\rho)^{-1}A^2,
$$
where the latter equality follows since $\cE(\rho)$ and $A$ commute. By the Schwarz
inequality we have
\begin{align*}
\|A\|_1^2&=(\Tr|A|)^2=\bigl(\Tr\cE(\rho)^{1/2}\cdot\cE(\rho)^{-1/2}|A|\bigr)^2 \\
&\le\Tr\cE(\rho)\cdot\Tr\cE(\rho)^{-1}A^2=\Tr\cE(\rho)^{-1}A^2.
\end{align*}
Therefore, $\|A\|_1^2\le\<A,\Omega_\rho^\kappa(A)\>$.
\end{proof}

\noindent{\bf Proof of Theorem~\ref{thm:trace-Riem}.}\enspace
Let $A\in\bH_d^0$ with $A\ne0$ and assume further that $A$ is invertible. Set
$\rho\equiv |A|/\|A\|_1\in\cD_d$. By the above lemma we have
$$
\|\Phi(A)\|_1^2\le\bigl\<\Phi(A),\Omega_{\Phi(\rho)}^\kappa(\Phi(A))\bigr\>.
$$
On the other hand, since $\rho$ and $A$ commute, we have
$$
\<A,\Omega_\rho^\kappa(A)\>=\Tr\rho^{-1}A^2=\|A\|_1\Tr|A|=\|A\|_1^2.
$$
Therefore,
$$
{\|\Phi(A)\|_1^2\over\|A\|_1^2}
\le{\bigl\<\Phi(A),\Omega_{\Phi(\rho)}^\kappa(\Phi(A))\bigr\>\over
\<A,\Omega_\rho^\kappa(A)\>}\le\eta_\kappa^\riem(\Phi).
$$
By continuity we have $\|\Phi(A)\|_1^2/\|A\|_1^2\le\eta_\kappa^\riem(\Phi)$ for all
$A\in\bH_d^0$ with $A\ne0$, proving the desired inequality. \qed


\subsection{QC and CQ channels}

Let $\Phi:\bM_d\to\bM_{d'}$ be a positive trace-preserving map. We call $\Phi$ a
{\it quantum-classical} (QC) channel if the range of $\Phi$ is included in a commutative
subalgebra of $\bM_{d'}$, and a {\it classical-quantum} (CQ) channel if the range of
$\widehat\Phi$ is in a commutative subalgebra of $\bM_d$. Note that if $\Phi$ is QC or CQ,
then positivity is the same asx CP, so it is indeed a channel. The following facts are easy
to see: $\Phi$ is a QC channel if and only if there are an orthonormal basis
$\{\psi_k\}_{k=1}^{d'}$ of $\bC^{d'}$ and a POVM $\{F_k\}_{k=1}^{d'}$ in $\bM_d$ such that
$$
\Phi(\rho)=\sum_k(\Tr F_k\rho)|\psi_k\>\<\psi_k|.
$$
Also, $\Phi$ is a CQ channel if and only if there are an orthonormal basis
$\{\phi_k\}_{k=1}^d$ of $\bC^d$ and density matrices $\gamma_k\in\cD_{d'}$, $1\le k\le d$,
such that
$$
\Phi(\rho)=\sum_k\<\phi_k,\rho\phi_k\>\gamma_k.
$$
Thus, our notions of QC and CQ channels coincide with those introduced in \cite{Ho}.
We note that both QC and CQ channels have the form 
$\Phi(\rho)=\sum_k(\Tr F_k\rho)\gamma_k$ with a POVM $\{F_k\}$ in $\bM_d$ a
 $\gamma_k\in\cD_{d'}$,  introduced in \cite{Ho} and shown to be
 {\em entanglement breaking}.   Several equivalent characterizations of
 this class were given in \cite{HSR}.

We remark that when $\Phi$ is purely classical, i.e., the ranges of $\Phi$ and
$\widehat\Phi$ are in commutative subalgebras of $\bM_{d'}$ and $\bM_d$, respectively,
 then $\Phi$ can be represented by a $d'\times d$ column-stochastic matrix in some orthonormal
bases of $\bC^d$ and $\bC^{d'}$.  Thus \cite[Theorem 1]{CRS} implies that
$$
\eta_\kappa^\riem(\Phi)=\eta_g^\relent(\Phi)
$$
for any independent choices of $\kappa\in\cK$ and $g\in\cG$.

\begin{prop}\label{prop:semi-classical}
Let $\kappa_1,\kappa_2\in\cK$ and assume that $\kappa_1(x)\le\kappa_2(x)$ for all $x>0$.
Then $\eta_{\kappa_1}^\riem(\Phi)\ge\eta_{\kappa_2}^\riem(\Phi)$ for every QC channel $\Phi$,
and $\eta_{\kappa_1}^\riem(\Phi)\le\eta_{\kappa_2}^\riem(\Phi)$ for every CQ channel $\Phi$.
\end{prop}

\begin{proof}
The assumption $\kappa_1\le\kappa_2$ implies that
$\Omega_\rho^{\kappa_1}\le\Omega_\rho^{\kappa_2}$ as operators on the Hilbert space
$\bM_d$ for every $\rho\in\cD_d$. When $\Phi$ is QC, we have
$$
\eta_\kappa^\riem(\Phi)=\sup_{\rho\in\cD_d}\,\sup_{A\in\bH_d^0,\,A\ne0}
{\Tr\Phi(\rho)^{-1}\Phi(A)^2\over\<A,\Omega_\rho^\kappa(A)\>}
$$
for every $\kappa\in\cK$. Hence $\eta_{\kappa_1}^\riem(\Phi)\ge\eta_{\kappa_2}^\riem(\Phi)$.
When $\Phi$ is CQ, choose a subalgebra $\cA$ of $\bM_d$ including the range of $\widehat\Phi$,
and let $\cE:\bM_d\to\cA$ be the trace-preserving conditional expectation. Since $\widehat\cE$
is nothing but the inclusion $\cA\hookrightarrow\bM_d$, we have
$\widehat\Phi=\widehat\cE\widehat\Phi$ so that $\Phi=\Phi\cE$. Therefore, we have
$$
\eta_\kappa^\riem(\Phi)=\sup_{\rho\in\cD_d\cap\cA}\,\sup_{A\in\bH_d^0\cap\cA,\,A\ne0}
{\<\Phi(A),\Omega_{\Phi(\rho)}^\kappa(\Phi(A))\>\over\Tr\rho^{-1}A^2}
$$
for every $\kappa\in\cK$. Hence
$\eta_{\kappa_1}^\riem(\Phi)\le\eta_{\kappa_2}^\riem(\Phi)$.
\end{proof}

Thanks to \eqref{min-k-max}, by Proposition \ref{prop:semi-classical} and Theorem
\ref{thm:Riem-trace} we have

\begin{cor}\label{cor:semi-classical}
If $\Phi$ is a QC channel and $\kappa\in\cK$ satisfies $\kappa(x)\ge x^{-1/2}$ for all $x>0$,
then
$$
\eta_{\max}^\riem(\Phi)\le\eta_\kappa^\riem(\Phi)
\le\eta_{x^{-1/2}}^\riem(\Phi)\le\eta^{\mathrm{Tr}}(\Phi).
$$
If $\Phi$ is a CQ channel and $\kappa\in\cK$ satisfies $\kappa(x)\le x^{-1/2}$ for all $x>0$,
then
$$
\eta_{\min}^\riem(\Phi)\le\eta_\kappa^\riem(\Phi)
\le\eta_{x^{-1/2}}^\riem(\Phi)\le\eta^{\mathrm{Tr}}(\Phi).
$$
\end{cor}

The next theorem shows a relation of $\eta_g^\relent(\Phi)$ for any $g\in\cG_\sym$ with
the coefficient $\eta_{\min}^\riem(\Phi)$ with respect to the minimal metric when $\Phi$
is a QC channel. The proof is based on the integral decomposition \eqref{int-exp2}.

\begin{thm}\label{thm:RelEnt-BU}
Assume that $\Phi:\bM_d\to\bM_{d'}$ is a QC channel. Then for every $\kappa\in\cK$ and
the corresponding $g\in\cG_\sym$ in \eqref{kg-rel},
$$
\eta_\kappa^\riem(\Phi)\le\eta_g^\relent(\Phi)\le\eta_{\min}^\riem(\Phi),
$$
and in particular,
\begin{equation}\label{min-riem-ent}
\eta_{\min}^\riem(\Phi)=\eta_{g_{\min}}^\relent(\Phi)
\end{equation}
for $g_{\min}\in\cG_\sym$ given in Example 5.
\end{thm}

\begin{proof}
The first inequality holds for a general CPT map $\Phi$ due to Theorem
\ref{thm:Riem-RelEnt}. Now we assume that $\Phi$ is a QC channel and let $\kappa\in\cK$ be
arbitrary. Since $\kappa$ admits an integral expression (or the extremal decomposition)
in \eqref{int-exp2} so that we write
$$
\kappa(x)=\int_{[0,1]}\kappa_s(x)\,dm(s),\qquad x\in(0,\infty)
$$
with $\kappa_s\in\cK$, $0\le s\le1$, given in \eqref{ext-k}. Moreover, for $0\le s\le1$ let
\begin{equation}\label{g-ext}
g_s(x)\equiv{(x-1)^2\over x+s}\in\cG,
\end{equation}
whose symmetrization $(g_s)_\sym\in\cG_\sym$ corresponds in \eqref{kg-rel} to $\kappa_s$, i.e.,
$(g_s)_\sym(x)=(x-1)^2\kappa_s(x)$. For every $\rho,\gamma\in\cD_d$ we then have
$$
H_g(\rho,\gamma)=\int_{[0,1]}H_{(g_s)_\sym}(\rho,\gamma)\,dm(s).
$$
So it suffices to prove that
\begin{equation}\label{H-BU1}
H_{(g_s)_\sym}(\Phi(\rho),\Phi(\gamma))
\le\eta_{\min}^\riem(\Phi)H_{(g_s)_\sym}(\rho,\gamma),\qquad0\le s\le1.
\end{equation}
Since
\begin{align}
R_\gamma^{-1}\kappa_s(L_\rho R_\gamma^{-1})
&={1+s\over2}\,R_\gamma^{-1}\biggl({1\over L_\rho R_\gamma^{-1}+s}
+{1\over sL_\rho R_\gamma^{-1}+1}\biggr) \nonumber\\
&={1+s\over2}\biggl({1\over L_\rho+sR_\gamma}
+{1\over sL_\rho+R_\gamma}\biggr), \label{L-R}
\end{align}
we have
\begin{align*}
&{R_\gamma^{-1}\kappa_s(L_\rho R_\gamma^{-1})
+R_\rho^{-1}\kappa_s(L_\gamma R_\rho^{-1})\over2} \\
&\qquad={1+s\over2}\biggl\{{(L_\rho+sR_\gamma)^{-1}+(sL_\gamma+R_\rho)^{-1}\over2}
+{(sL_\rho+R_\gamma)^{-1}+(L_\gamma+sR_\rho)^{-1}\over2}\biggr\} \\
&\qquad\ge{1+s\over2}\biggl\{\biggl({L_\rho+sR_\gamma+sL_\gamma+R_\rho\over2}\biggr)^{-1}
+\biggl({sL_\rho+R_\gamma+L_\gamma+sR_\rho\over2}\biggr)^{-1}\biggr\} \\
&\qquad={1\over2}\Biggl\{\Bigg({L_{\rho+s\gamma\over1+s}
+R_{\rho+s\gamma\over1+s}\over2}\Biggr)^{-1}+\Biggl({L_{\gamma+s\rho\over1+s}
+R_{\gamma+s\rho\over1+s}\over2}\Biggr)^{-1}\Biggr\}.
\end{align*}
In the above the operator convexity of $x^{-1}$ on $(0,\infty)$ has been used. We then have
\begin{align}
H_{(g_s)_\sym}(\rho,\gamma)
&={H_{(g_s)_\sym}(\rho,\gamma)+H_{(g_s)_\sym}(\gamma,\rho)\over2} \nonumber\\
&=\Biggl\<\rho-\gamma,\Biggl\{{R_\gamma^{-1}\kappa_s(L_\rho R_\gamma^{-1})
+R_\rho^{-1}\kappa_s(L_\gamma R_\rho^{-1})\over2}\Biggr\}(\rho-\gamma)\Biggr\> \nonumber\\
&\ge{1\over2}\Biggl\<(\rho-\gamma,\Biggl\{\Biggl({L_{\rho+s\gamma\over1+s}
+R_{\rho+s\gamma\over1+s}\over2}\Biggr)^{-1}+\Biggl({L_{\gamma+s\rho\over1+s}
+R_{\gamma+s\rho\over1+s}\over2}\Biggr)^{-1}\Biggr\}(\rho-\gamma)\Biggr\> \nonumber\\
&={1\over2}\biggl\<\rho-\gamma,\biggl\{\Omega_{\rho+s\gamma\over1+s}^{\min}
+\Omega_{\gamma+s\rho\over1+s}^{\min}\biggr\}(\rho-\gamma)\biggr\> \label{H-BU2}
\end{align}
thanks to \eqref{Omega:min}.

Now we consider the case where $\rho$ and $\gamma$ commute. Since $\rho$ and $\gamma$
commute with $\rho-\gamma$, it is easy to see from \eqref{L-R} that
\begin{align*}
H_{(g_s)_\sym}(\rho,\gamma)
&=\bigl\<\rho-\gamma,R_\gamma^{-1}\kappa_s(L_\rho R_\gamma^{-1})(\rho-\gamma)\bigr\> \\
&={1+s\over2}\bigl\<\rho-\gamma,
\{(\rho+s\gamma)^{-1}+(\gamma+s\rho)^{-1}\}(\rho-\gamma)\bigr\> \\
&={1\over2}\biggl\<\rho-\gamma,\biggl\{\Omega_{\rho+s\gamma\over1+s}^{\min}
+\Omega_{\gamma+s\rho\over1+s}^{\min}\biggr\}(\rho-\gamma)\biggr\>.
\end{align*}
Since $\Phi$ has the commutative range, we can apply the above to $\Phi(\rho)$ and
$\Phi(\gamma)$ to obtain
\begin{equation}\label{H-BU3}
H_{(g_s)_\sym}(\Phi(\rho),\Phi(\gamma))={1\over2}\biggl\<\Phi(\rho-\gamma),\biggl\{
\Omega_{\Phi\bigl({\rho+s\gamma\over1+s}\bigr)}^{\min}
+\Omega_{\Phi\bigl({\gamma+s\rho\over1+s}\bigr)}^{\min}\biggr\}\Phi(\rho-\gamma)\biggr\>.
\end{equation}
Hence \eqref{H-BU1} follows from \eqref{H-BU2} and \eqref{H-BU3}.
\end{proof}

\subsection{Weak Schwarz maps}

Although we have restricted  our consideration to quantum channels 
 in the usual sense of CPT maps, the
monotonicity property of $g$-divergences and monotone metrics holds more generally under a
positive trace-preserving map $\Phi:\bM_d\to\bM_{d'}$ whose adjoint $\widehat\Phi$ is a
{\it weak} Schwarz map in the sense that $\widehat\Phi(Y^*)\widehat\Phi(Y)\le\widehat\Phi(Y^*Y)$
for all $Y\in\bM_{d^\prime}$. The proofs of monotonicity of $g$-divergences in \cite{Pe1, Pe2}
as well as the argument  in \cite{Pe3} for monotone metrics requires only this
weaker condition, as discussed further in \cite{HMPB}.
 Thus, we can define the contraction coefficients for such maps $\Phi$ rather
than CPT maps, and most results in the paper extend to this slightly more general situation.

We introduce the term   ``weak Schwarz map" for the following reason.
For a positive linear functional $\phi$ on an operator algebra, the Schwarz inequality can be
written as $|\phi(A^*B)|^2\le\phi(A^*A)\phi(B^*B)$. The analogous result for a linear map
$\Phi$ on an operator algebra is the {\it operator} inequality
\be \label{Schwarz}
\Phi(A^*B)[\Phi(B^*B)]^{-1}\Phi(B^*A)\le\Phi(A^*A)
\ee
first proved for CP maps by  Lieb and Ruskai \cite{LiRu} in 1974. (In   finite
dimensions, an equivalent inequality was proved much earlier by Kiefer \cite{Ki}  in 1959.)
In 1980 Choi \cite{Ch2} showed that \eqref{Schwarz} holds if and only if $\Phi$
is $2$-positive. Thus, it would be natural to consider the Schwarz maps as precisely the class
of $2$-positive maps. However, earlier in 1952 Kadison \cite{Ka} proved a special case of
\eqref{Schwarz} with $A=A^*$ and $B=I$, and it was later found that the condition
$A = A^*$ could be dropped in many situations.  Thus, the terms
``Schwarz inequality" and ``Schwarz map" were associated with the weaker inequality
$\Phi(A^*)\Phi(A)\le\Phi(A^*A)$.   However, this  inequality does not hold for arbitrary positive 
linear maps and we know of no characterization of the subclass for which it holds other than the inequality itself.

\section{Qubit Channels}

We now consider some special qubit channels.  Some of these results were stated
without proof at the end of \cite{LR}.  Others are new and resolve conjectures
discussed elsewhere in the paper \cite{LR}.

We first recall the description of $\cD_2$ as the Bloch ball briefly, see, e.g.,
\cite{NC,Pe4,RSW} for more details. Pauli matrices
$\sigma_1=\begin{pmatrix}0&1\\1&0\end{pmatrix}$,
$\sigma_2=\begin{pmatrix}0&-i\\i&0\end{pmatrix}$ and
$\sigma_3=\begin{pmatrix}1&0\\0&-1\end{pmatrix}$ with $I=\begin{pmatrix}1&0\\0&1\end{pmatrix}$
form an orthogonal basis of the qubit Hilbert space $\bM_2$. Any $\rho\in\cD_2$ is represented
as $\rho={1\over2}[I+\bw\dtsig]$ by a unique $\bw=(w_1,w_2,w_3)^t\in\bR^3$ with
$|\bw|\equiv\sqrt{w_1^2+w_2^2+w_3^2}\le1$, where
$\bw\dtsig\equiv w_1\sigma_1+w_2\sigma_2+w_3\sigma_3$. Here, $\rho$ is pure if and only if
$\bw$ is on the unit sphere, i.e., $|\bw|=1$. A trace-preserving linear map $\Phi:\bM_2\to\bM_2$
is represented as
$$
\Phi(w_0I+\bw\dtsig)=w_0I+(w_0\bt+T\bw)\dtsig,\qquad w_0\in\bR,\ \bw\in\bR^3,
$$
by a vector $\bt\in\bR^3$ and a $3\times3$ real matrix $T=(t_{ij})_{i,j=1}^3$, which,
 as observed in \cite{KR},  can be assumed to be diagonal without loss of generality.   
 Clearly, $\Phi$ is
positive if and only if $|\bt+T\bw|\le1$ for all $\bw\in\bR^3$ with $|\bw|\le1$, and $\Phi$ is
unital and positive if and only if $\bt=0$ and $\|T\|_\infty\le1$, where $\|T\|_\infty$ is the
operator norm of $T$.  Necessary and sufficient conditions for complete positivity
 were given in \cite{RSW}.    In the special case when only $t_3 \neq 0$, it was
shown earlier by  Fujiwara  and Algoet  \cite{FA} that  a map of this form is CPT if and only if
$$
(\lambda_1 \pm\lambda_2 )^2 \leq (1 \pm \lambda_3)^2 -  t_3^2
$$
when $\lambda_k $ are the diagonal elements of $T$.

\begin{thm}\label{thm:unital}
For any unital map $\Phi_{T}:I + \bw \dtsig \mapsto I + (T\bw)\dtsig $ where $T$ is a real
matrix with $\norm{T}_\infty \leq 1$,
\begin{eqnarray*} 
\eta_\kappa^{\riem}(\Phi_{T}) = \eta_\kappa^{\geod}(\Phi_{T}) = \eta_g^{\RE}(\Phi_{T})
= \| T \|_\infty ^2
\end{eqnarray*}
for every $\kappa\in\cK$ and every $g\in\cG$. Furthermore, $\eta^{\tr}(\Phi_{T})=\|T\|_\infty$.
\end{thm}


This result does not require that the map $\Phi_T$ be CP. Note that
$\eta^{\tr}(\Phi_{T}) = \|T\|_\infty=\sqrt{\eta_\kappa^{\riem}(\Phi_{T})}$, which is consistent
with Theorem \ref{thm:trace-Riem}.

The next theorem treats a family of trace-preserving maps $\Phi_{\alpha,\tau}:\bM_2\to\bM_2$
with two real parameters $\alpha,\tau$ determined by
$\bt=(0,0,\tau)^t$ and $T=\diag(\alpha,0,0)$; more explicitly,
$$
\Phi_{\alpha,\tau}(w_0I+\bw\dtsig)=w_0I+\alpha w_1\sigma_1+\tau w_0\sigma_3.
$$
In this case, the condition $\alpha^2+\tau^2\le1$ is necessary and sufficient for both
 positivity and complete positivity  
 as shown in \cite{FA,RSW}. It is also easy to check that the range of the
adjoint map $\widehat\Phi$ is included in the commutative subalgebra generated by
$\{I,\sigma_1\}$. Thus,  when $\tau \neq 0 $,
$\Phi_{\alpha,\tau}$ is a non-unital CQ channel.
 Below we assume that $\alpha\ge0$ and  $\alpha^2+\tau^2\le1$.

\begin{thm}\label{thm:etaCQ}
Let $\Phi \equiv \Phi_{\alpha,\tau}$ be a non-unital CQ channel with $\alpha,\tau$ specified
above. Then
\bsq \begin{eqnarray}
\eta^{\tr}(\Phi) & = & \alpha, \label{etaCQ:a}\\
\eta_{\max}^\riem(\Phi) & = & \frac{\alpha^2}{1 - \tau^2}, \label{etaCQ:b}\\
\eta_{{\rm \wh{WY} }}^\riem(\Phi) \equiv \eta_{(1+\sqrt x)^2/4x}^\riem(\Phi)
& \geq & \alpha^2\,\frac{ 1 + \sqrt{1 - \tau^2}}{2 (1 - \tau^2)}, \label{etaCQ:c}\\
\eta_{ x^{-1/2}}^\riem(\Phi) & \geq & \frac{\alpha^2}{\sqrt{1-\tau^2}}, \label{etaCQ:d}\\
\eta_{{\rm BKM}}^\riem(\Phi) \equiv \eta_{(\log x)/(x-1)}^\riem(\Phi)
& \geq & \frac{\alpha^2 }{2\tau} \log \tfrac{ 1 + \tau }{ 1 -\tau}, \label{etaCQ:e}\\
\eta_{{\rm WY}}^\riem(\Phi) \equiv \eta_{4/(1+\sqrt x)^2}^\riem(\Phi)
& = & \frac{2\alpha^2}{1 + \sqrt{1-\tau^2}}, \label{etaCQ:f}\\
\eta^{\riem}_{\min}(\Phi) & = & \alpha^2. \label{etaCQ:g}
\end{eqnarray} \esq 
Moreover, for the extreme points $\kappa_s$ of $\cK$ given in \eqref{ext-k},
\begin{equation}\label{etaCQ:ext}
\eta^{\riem}_{\kappa_s}(\Phi) = 
\frac{\alpha^2}{1 - \left( \frac{1-s}{1+s} \right)^2 \tau^2},\qquad 0\le s\le1.
\end{equation}
\end{thm}

In the above, the function $(1+\sqrt x)^2/4x\in\cK$ is the dual of $\kappa_\WY$ in Example 4,
i.e., $1/\kappa_\WY(x^{-1})=(1+\sqrt x)^2/4x$. The identities \eqref{etaCQ:b} and
\eqref{etaCQ:g} are of course contained in \eqref{etaCQ:ext}. By Proposition
\ref{prop:semi-classical} (for CQ channels) together with \eqref{etaCQ:b} and \eqref{etaCQ:g}
we observe that
\be \label{qubitCQgen}
\alpha^2 \leq \eta_\kappa^{\riem}(\Phi_{\alpha,\tau}) \leq \frac{\alpha^2}{1 - \tau^2}
\ee
for every $\kappa\in\cK$.

Although the bounds in the above theorem are sufficient to disprove two conjectures as remarked
below, we believe that they are optimal, i.e.,

\begin{conj}
Equality holds in \eqref{etaCQ:c} through \eqref{etaCQ:e} above.
\end{conj}
In those cases in which we can compute  $ \eta_{\kappa}^\riem(\Phi_{\alpha,\tau}) $ exactly,
the supremum is attained when $\rho = \half I$ or, equivalently, $\bw = (0,0,0)^t$ and
 $A=\by\dtsig$ with $\by = (y_1, 0, 0)^t$.
Since the output of $\Phi$ does not involve $w_2, w_3$ it is reasonable that there is no loss
of generality in choosing $w_2 = w_3 = 0$. And since the channel is symmetric around $w_1 = 0$
or $P =  \half I$, this choice is also reasonable.  But a proof for arbitrary choices of $\kappa$ does not seem easy.

If  the above conjecture is true, then for these examples $\kappa_1 \leq \kappa_2$
implies $\eta_{\kappa_1}^\riem(\Phi_{\alpha,\tau})\leq\eta_{\kappa_2}^\riem(\Phi_{\alpha,\tau})$,
which is consistent with Proposition \ref{prop:semi-classical}.

\begin{remark}\rm
It follows from \eqref{etaCQ:a} and \eqref{etaCQ:b} that the conjecture \cite{LR} that 
$\eta^{\riem}_g(\Phi) \leq \eta^{\tr}(\Phi)$ is false.  Indeed, whenever $\alpha > 1 - \tau^2$,
it follows that  
\be \label{conjdobfalse}
\eta_{\max}^\riem(\Phi) = \frac{\alpha^2}{1 - \tau^2} > \alpha = \eta^{\tr}(\Phi). 
\ee
Since $\alpha > \alpha^2$, parameters can be found that are consistent with the CP condition
but satisfy \eqref{conjdobfalse}. In fact, $\alpha = \tau = 1/\sqrt2$ will do.
\end{remark}

\begin{remark}\rm
Although we do not know when equality holds in the bounds above, we
have sufficient information to conclude that the largest contraction coefficient is not
necessarily given by $\kappa(x) = x^{-1/2}$ as conjectured in \cite{KT}. In particular, when
$\alpha^2 + \tau^2 = 1$, the bound $\eta_{x^{-1/2}}^\riem(\Phi) \leq \eta^{\tr}(\Phi)$
(Theorem \ref{thm:Riem-trace}) implies that equality holds in \eqref{etaCQ:d}, so
$\eta_{x^{-1/2}}^\riem(\Phi)=\alpha<1=\eta_{\max}^\riem(\Phi)$ in this case.
This is foreshadowed by the fact that  Proposition \ref{prop:semi-classical} 
says that $\eta_\kappa^\riem(\Phi)$ is monotone
increasing in $\kappa\in\cK$ for a CQ channel.

 We also note that  the inequality
$\eta_{x^{-1/2}}^{\riem} \leq \eta^{\tr}$ becomes
$\frac{\alpha^2}{ \sqrt{1- \tau^2}} \leq \alpha$ or $\frac{\alpha}{ \sqrt{1- \tau^2}} \leq 1$,
which is equivalent to the CP condition $\alpha^2 + \tau^2 \leq 1 $.   

\end{remark}
The next theorem disproves the conjecture \cite{LR} that
$\eta_\kappa^\riem(\Phi)=\eta_g^\relent(\Phi)$ for every $\kappa\in\cK$ and the corresponding
$g\in\cG_\sym$. For $0\le s\le1$ let $g_s\in\cG$ be given as in \eqref{g-ext}, so
$(g_s)_\sym\in\cG_\sym$ corresponds to to the extreme $\kappa_s\in\cK$.

\begin{thm} \label{thm:relent.neq.riem}
Let $\Phi_{\alpha,\tau}$ be as above. If $4\tau^2>(1-\alpha^2)(4-\alpha^2)$ (this is the case
when $\alpha^2+\tau^2=1$ with $\alpha>0$), then
\begin{equation}\label{relent-ext}
\eta_{g_s}^\relent(\Phi_{\alpha,\tau})\ge\eta_{(g_s)_\sym}^\relent(\Phi_{\alpha,\tau})
>\eta_{\kappa_s}^\riem(\Phi_{\alpha,\tau})
\end{equation}
for any $s\le1$ sufficiently near $1$ (depending on $\alpha,\tau$). Moreover, if
$s>\sqrt{\tfrac{5-\sqrt{21}}{2}}\approx0.457$, then
$\eta_{(g_s)_\sym}^\relent(\Phi_{\alpha,\tau})>\eta_{\kappa_s}^\riem(\Phi_{\alpha,\tau})$ when
$\alpha^2=1-\tau^2$ is sufficiently small.
\end{thm}

Note that $(g_0)_\sym$ is $g_{\max}$ given in Example 2 and the three coefficients in
\eqref{relent-ext} are equal in particular when $s=0$, as will be shown in Proposition
\ref{prop:max} for general CPT maps.

Theorems \ref{thm:unital}, \ref{thm:etaCQ} and \ref{thm:relent.neq.riem} will be proved in
the whole Appendix B.

\section{Results in Special Cases}\label{sect:special-cases}


\subsection{BKM metric}

\begin{thm}\label{thm:BKM}
For every CPT map $\Phi$,
$\eta_\BKM^\riem(\Phi)=\eta_\BKM^\relent(\Phi)$.
\end{thm}

\begin{proof}
By Theorem \ref{thm:Riem-RelEnt} it suffices to prove that
\begin{equation}\label{ineq:BKM}
\eta_\BKM^\relent(\Phi)\le\eta_\BKM^\riem(\Phi)
\end{equation}
for every CPT map $\Phi$. To do this, consider the line segment
$\xi(t)\equiv (1-t)\rho+t\gamma=\rho+t(\gamma-\rho)$, $0\le t\le1$, joining
$\rho,\gamma\in\cD_d$. By using Daleckii and Krein's differential formula (see, e.g.,
\cite[Section 2.3]{Hi}) we compute the derivative
\begin{align*}
{d\over dt}\,H(\rho,\xi(t))&={d\over dt}\,\Tr\rho(\log\rho-\log\xi(t))
=-\Tr\rho\biggl({d\over dt}\log\xi(t)\biggr) \\
&=-\Tr\rho\log^{[1]}(L_{\xi(t)},R_{\xi(t)})(\gamma-\rho),
\end{align*}
where $\log^{[1]}(x,y)\equiv (\log x-\log y)/(x-y)$, the divided difference of $\log x$. We
hence have
\begin{equation}\label{derivative1}
{d\over dt}\,H(\rho,\xi(t))
=-\Tr\rho\,{\log L_{\xi(t)}-\log R_{\xi(t)}\over L_{\xi(t)}-R_{\xi(t)}}\,(\gamma-\rho)
=-\bigl\<\rho,\Omega_{\xi(t)}^\BKM(\gamma-\rho)\bigr\>
\end{equation}
thanks to \eqref{Omega:BKM}, and similarly
\begin{equation}\label{derivative2}
{d\over dt}\,H(\gamma,\xi(t))
=-\bigl\<\gamma,\Omega_{\xi(t)}^\BKM(\gamma-\rho)\bigr\>.
\end{equation}
Therefore,
$$
{d\over dt}\,\{H(\rho,\xi(t))-H(\gamma,\xi(t))\}
=\bigl\<\gamma-\rho,\Omega_{\xi(t)}^\BKM(\gamma-\rho)\bigr\>
$$
so that
\begin{align*}
H_\BKM(\rho,\gamma)&\equiv H(\rho,\gamma)+H(\gamma,\rho) \\
&=\{H(\rho,\xi(1))-H(\gamma,\xi(1))\}-\{H(\rho,\xi(0))-H(\gamma,\xi(0))\} \\
&=\int_0^1{d\over dt}\,\{H(\rho,\xi(t))-H(\gamma,\xi(t))\}\,dt \\
&=\int_0^1\bigl\<\gamma-\rho,\Omega_{\xi(t)}^\BKM(\gamma-\rho)\bigr\>\,dt.
\end{align*}
By replacing $\rho,\gamma$ with $\Phi(\rho),\Phi(\gamma)$ we also have
$$
H_\BKM(\Phi(\rho),\Phi(\gamma))
=\int_0^1\bigl\<\Phi(\gamma-\rho),
\Omega_{\Phi(\xi(t))}^\BKM(\Phi(\gamma-\rho))\bigr\>\,dt.
$$
Since
$$
\bigl\<\Phi(\gamma-\rho),
\Omega_{\Phi(\xi(t))}^\BKM(\Phi(\gamma-\rho))\bigr\>
\le\eta_\BKM^\riem(\Phi)\bigl\<\gamma-\rho,
\Omega_{\xi(t)}^\BKM(\gamma-\rho)\bigr\>,\qquad0\le t\le1,
$$
the desired inequality \eqref{ineq:BKM} follows.
\end{proof}
 
The differential expressions in \eqref{derivative1} and \eqref{derivative2} are quite special,
so it seems that we cannot apply the differential method as above for other
$\kappa\in\cK$. One may also consider the contraction coefficient  defined in
\eqref{eta-relent}
$$
\eta_{x\log x}^\relent(\Phi)\equiv \sup_{\rho,\gamma\in\cD_d,\,\rho\ne\gamma}
{H(\Phi(\rho),\Phi(\gamma))\over H(\rho,\gamma)}
$$
with respect to the standard (non-symmetrized) relative entropy.
One has $\eta_\BKM^\relent(\Phi)\le\eta_{x\log x}^\relent(\Phi)$ by the inequality in
\eqref{eta-g-sym}, but it is unknown whether both contraction coefficients coincide or not.

\subsection{Other special results}

We collect here some additional special relations that may be of interest.

%
Since the maximal metric has the special property \eqref{quadspec} that every 
  metric $\<A,\Omega_\rho^{\max}(A)\>$ can be realized as a quadratic relative
entropy,  Theorem \ref{thm:Riem-RelEnt} immediately implies the following:
\begin{prop}\label{prop:max}
For every CPT map $\Phi$,
\begin{equation}\label{max-riem-ent}
\eta_{\max}^\riem(\Phi)=\eta_{g_{\max}}^\relent(\Phi)
=\eta_{(x-1)^2}^\relent(\Phi),
\end{equation}
where $g_{\max}\in\cG_\sym$ is given in Example \ref{ex:max}.
\end{prop}

The identities \eqref{max-riem-ent} and \eqref{min-riem-ent} show some asymmetry between the
contraction properties of $\kmax$ and $\kmin$; \eqref{max-riem-ent} holds for all quantum
channels while \eqref{min-riem-ent} does for only QC channels, see a counter-example ($s=1$)
in Theorem \ref{thm:relent.neq.riem} for a CQ channel.


The functions $\kappa_t^\WYD$ given in \eqref{k-WYD} showed up through the representation of
the Wigner-Yanase-Dyson skew information in terms of monotone metrics, as described in
\cite[Section 2.4, Example 4.8]{HKPR}. Furthermore, as studied in \cite{JR}, the trace
functional of WYD concavity/convexity \cite{Li,An2} is recovered by the quasi-entropy for
$g^{(t)}$ given in \eqref{g^t} as follows:
$$
J_t(K,A,B)\equiv \<KB^{1/2},g^{(t)}(L_AR_B^{-1})KB^{1/2}\>
={1\over t(1-t)}\,(\Tr K^*AK-\Tr K^*A^tKB^{1-t})
$$
for $A,B\in\bP_d$ with a linear term. The $g^{(t)}$-divergence is
$$
H_t(\rho,\gamma)\equiv J_t(I,\rho,\gamma)={1-\Tr\rho^t\gamma^{1-t}\over t(1-t)}.
$$
Note that $H(\rho,\gamma)=\lim_{t\to1}H_t(\rho,\gamma)$ and
$H(\gamma,\rho)=\lim_{t\to0}H_t(\rho,\gamma)$ so that $H_t(\rho,\gamma)$ forms a
one-parameter extension of the relative entropy. By Theorem \ref{thm:Riem-RelEnt} we have

\begin{prop}\label{prop:WYD}
For every CPT map $\Phi$ and every $t\in(0,1)$,
$$
\eta_{\kappa_t^\WYD}^\riem(\Phi)\le\eta_{g^{(t)}}^\relent(\Phi)
=\sup_{\rho,\gamma\in\cD_d,\,\rho\ne\gamma}
{1-\Tr\Phi(\rho)^t\Phi(\gamma)^{1-t}\over1-\Tr\rho^t\gamma^{1-t}}.
$$
\end{prop}

For $\kappa_\WY(x)=4/(1+\sqrt x)^2$ and $\kmin(x)=2/(1+x)$, Theorem \ref{thm:Riem-geod} together
with \eqref{geod-WY} and \eqref{geod-Bures} yields
\begin{align}
\eta_\WY^\riem(\Phi)&=\sup_{\rho,\gamma\in\cD_d,\,\rho\ne\gamma}
\Biggl[{\arccos\Tr\Phi(\rho)^{1/2}\Phi(\gamma)^{1/2}\over
\arccos\Tr\rho^{1/2}\gamma^{1/2}}\Biggr]^2, \nonumber\\
\eta_{\min}^\riem(\Phi)&=\sup_{\rho,\gamma\in\cD_d,\,\rho\ne\gamma}
\biggl[{\arccos F(\Phi(\rho),\Phi(\gamma))\over\arccos F(\rho,\gamma)}\biggr]^2.
\label{Riem:min}
\end{align}
Since
$$
\biggl({\arccos t\over\arccos s}\biggr)^2<{1-t\over1-s}\quad
\mbox{for $0<s<t<1$},
$$
from \eqref{Riem:min} and \eqref{dist-Bures} we also have

\begin{prop}\label{prop:min}
For every CPT map $\Phi$,
\begin{align*}
\eta_{\min}^\riem(\Phi)&\le\sup_{\rho,\gamma\in\cD_d,\,\rho\ne\gamma}
{1-F(\Phi(\rho),\Phi(\gamma))\over1-F(\rho,\gamma)}
=\sup_{\rho,\gamma\in\cD_d,\,\rho\ne\gamma}
\biggl[{d_\Bures(\Phi(\rho),\Phi(\gamma))\over d_\Bures(\rho,\gamma)}\biggr]^2.
\end{align*}
\end{prop}

\medskip

\noindent{\bf Acknowledgments:}  Part of this work was done when both authors participated
in a workshop at  Centro de Ciencias de Benasque Pedro Pascual  in Benasque, Spain 
and during the ICM Satellite Conference on Operator Algebras and Applications
in Cheongpung, Korea.
The work of FH was supported in part by Grant-in-Aid for Scientific Research (C)21540208.
The work of MBR was partially supported by NSF grant CCF 1018401 which was administered by Tufts University.

\appendix

\section{Hiai-Petz Lemma}\label{sect:appendA}

The next lemma was proved in \cite{HP0} in a slightly different setting of the Riemannian
manifold $\bP_d$ and its tangent space $\bH_d$ instead of $\cD_d$ and $\bH_d^0$ here,
whose proof can work in the present setting as well. The proof is provided below for
completeness.

\begin{lemma}\label{lemma:limit-form}{\bf\cite[Lemma 4.2]{HP0}}\enspace
For every $\kappa\in\cK$, $\rho\in\cD_d$ and $A\in\bH_d^0$,
$$
\lim_{\eps\searrow0}{D_\kappa(\rho,\rho+\eps A)\over\eps}
=\sqrt{\<A,\Omega_\rho^\kappa(A)\>}.
$$
\end{lemma}

\begin{proof}
First, recall that if $\mathbb{T}$ is a linear operator on the Hilbert space
$(\bM_d,\<\cdot,\cdot\>)$ represented as the Schur multiplication by a matrix
$(t_{ij})\in \bH_d$, then $\mathbb{T}\ge0$ if and only if $t_{ij}\ge0$ for all
$i,j=1,\dots,n$. We denote by $\bI$ the identity operator on $(\bM_d,\<\cdot,\cdot\>)$,
which is represented as the Schur multiplication by the matrix of all entries equal to $1$.
To prove the lemma, we may assume that $\rho=\diag(\lambda_1,\dots,\lambda_d)$. For brevity
let $a_{ij}\equiv \lambda_j^{-1}\kappa(\lambda_i\lambda_j^{-1})$ for $i,j=1,\dots,d$ and
$\alpha\equiv \min_{i,j}a_{ij}>0$. Since
$$
\Omega_\rho^\kappa(X)=(a_{ij})_{ij}\circ X,\qquad X\in\bM_d
$$
(see, e.g., \cite[Section 2.2]{HKPR}), it follows that
$\Omega_\rho^\kappa\ge\alpha\bI$ as operators on $(\bM_d,\<\cdot,\cdot\>)$. For each
$\delta\in(0,\alpha)$, since $\gamma\in\cD_d\mapsto\Omega_\gamma^\kappa$ is continuous,
there exists an $r_1>0$ such that if $\gamma\in\cD_d$ and $\|\gamma-\rho\|_2<r_1$ then
$\|\Omega_\gamma^\kappa-\Omega_\rho^\kappa\|_\infty<\delta$, where $\|\cdot\|_\infty$
denotes the operator norm for operators on $(\bM_d,\<\cdot,\cdot\>)$. Furthermore,
since $D_\kappa$ and $\|\cdot\|_2$ define the same topology on $\cD_d$ (see
\cite[Chapter IV, Proposition 3.5]{KN}), there exists an $r_0>0$ such that if
$\gamma\in\cD_d$ and $D_\kappa(\gamma,\rho)<r_0$ then $\|\gamma-\rho\|_2<r_1$.

Now let $A\in\bH_d^0$ and choose a sufficiently small $\eps>0$ so that
$D_\kappa(\rho,\rho+\eps A)<r_0$ and $\eps\|A\|_2<r_1$. Let $\xi:[0,1]\to\cD_d$ be any
smooth curve from $\rho$ to $\rho+\eps A$ such that $L_\kappa(\xi)<r_0$, where
$L_\kappa(\xi)$ is the length of $\xi$ with respect to the monotone metric induced by
$\kappa$. Since $D_\kappa(\xi(t),\rho)<r_0$ and so $\|\xi(t)-\rho\|_2<r_1$ for all
$0\le t\le1$, we have
\begin{align*}
L_\kappa(\xi)
&=\int_0^1\sqrt{\bigl\<\xi'(t),\Omega_{\xi(t)}^\kappa\xi'(t)\bigr\>}\,dt \\
&\ge\int_0^1\sqrt{\<\xi'(t),(\Omega_\rho^\kappa-\delta\bI)\xi'(t)\>}\,dt \\
&=\int_0^1\|(\Omega_\rho^\kappa-\delta\bI)^{1/2}\xi'(t)\|\,dt \\
&\ge\|(\Omega_\rho^\kappa-\delta\bI)^{1/2}(\eps A)\|_2 \\
&=\eps\Big\|\Bigl((a_{ij}-\delta)^{1/2}\Bigr)_{ij}\circ A\Big\|_2.
\end{align*}
In the above, note that $\Omega_\rho^\kappa-\delta\bI\ge0$ on $(\bM_d,\<\cdot,\cdot\>)$
since $\delta<\alpha$. Also, the second inequality above follows since
$\int_0^1\|(\Omega_\rho^\kappa-\delta I)^{1/2}\xi'(t)\|\,dt$ is the length of the curve
$(\Omega_\rho^\kappa-\delta\bI)^{1/2}\xi(t)$, $0\le t\le1$, from
$(\Omega_\rho^\kappa-\delta\bI)^{1/2}\rho$ to
$(\Omega_\rho^\kappa-\delta\bI)^{1/2}(\rho+\eps A)$ in the Euclidean space
$(\bH_d,\|\cdot\|_2)$ and it is shortest if $\xi$ is the segment between $\rho$ and
$\rho+\eps A$. Taking the infimum of $L_\kappa(\xi)$ yields
$$
D_\kappa(\rho,\rho+\eps A)
\ge\eps\Big\|\Bigl((a_{ij}-\delta)^{1/2}\Bigr)_{ij}\circ A\Big\|_2.
$$
On the other hand, let $\xi_0(t)\equiv\rho+t\eps A$. Since
$\|\xi_0(t)-\rho\|_2\le\eps\|A\|_2<r_1$ for $0\le t\le1$, we have
\begin{align*}
D_\kappa(\rho,\rho+\eps A)&\le L_\kappa(\xi_0)
=\int_0^1\sqrt{\bigl\<\xi_0'(t),\Omega_{\xi_0(t)}^\kappa(\xi_0'(t))\bigr\>}\,dt \\
&\le\int_0^1\sqrt{\bigl\<\xi_0'(t),(\Omega_\rho^\kappa+\delta\bI)(\xi_0'(t))\bigr\>}\,dt \\
&=\|(\Omega_\rho^\kappa+\delta\bI)^{1/2}(\eps A)\|_2 \\
&=\eps\Big\|\Bigl((a_{ij}+\delta)^{1/2}\Bigr)_{ij}\circ A\Big\|_2.
\end{align*}
Since $\delta$ is arbitrary,
$$
\lim_{\eps\searrow0}{D_\kappa(\rho,\rho+\eps A)\over\eps}
=\Big\|\Bigl(a_{ij}^{1/2}\Bigr)_{ij}\circ A\Big\|_2
=\|(\Omega_\rho^\kappa)^{1/2}A\|_2=\sqrt{\<A,\Omega_\rho^\kappa(A)\>},
$$
as desired.
\end{proof}

\section{Qubit Proofs}    \label{sect:qubitpf}

\subsection{Unital qubit channels: Proof of Theorem~\ref{thm:unital}}

Writing $E = \half (I + \bw \dtsig)$ and $F = \half (I + \bx \dtsig)$ with $|\bw| = |\bx| = 1$,
we find $EF = 0 \Leftrightarrow \bw = -\bx$, in which case $E - F = \bw \dtsig$ so that by
  \eqref{Ruskai-form},
\be  \label{dob.unit}
\eta^{\tr}(\Phi_T) = \half \sup_{|\bw| = 1} \, \tr \, | (T \bw)\dtsig |
= \sup_{|\bw| = 1}|T\bw|=\norm{T}_\infty.
\ee
Note that this implies that $\Phi_T$ is ``non-scrambling'' if and only if $\| T \|_\infty = 1$. 

It then follows from Theorems~\ref{thm:Riem-geod} and \ref{thm:trace-Riem} that
$$
\eta_\kappa^\geod(\Phi_T) = \eta_\kappa^\riem(\Phi_T) \geq \norm{T}_\infty^2
$$
for all $\kappa\in\cK$. Thus it suffices to show that $\eta_g^\RE(\Phi_T)\leq\norm{T}_\infty^2$
for all $g \in \cG$. In fact, from the integral expression \eqref{H-int} it suffices to do this
for $g_s\equiv (x-1)^2/(x+s)$, $s\ge0$, as in \eqref{g-ext}, for which we have
\be \label{H-gs}
H_{g_s}(\rho,\gamma) = \Tr (\rho-\gamma) \frac{1}{L_\rho+sR_\gamma} (\rho-\gamma).
\ee
Let $\rho = \half(I + \bw \dtsig)$ and $\gamma = \half(I + \bx \dtsig)$ with $|\bw|,|\bx|<1$
and $\by\equiv\bw-\bx\ne0$ (which guarantee $\rho,\gamma\in\cD_2$ and $\rho\ne\gamma$).
For $s=0$ use \eqref{matprod} and \eqref{matinv} below to have
$$
H_{g_0}(\rho,\gamma)=2\Tr(\by\dtsig)^2(I+\bw\dtsig)^{-1}={4|\by|^2\over1-|\bw|^2}.
$$
Since $\Phi_T(\rho)=\half[I + (T\bw) \dtsig]$ and $\Phi_T(\by\dtsig)=(T\by)\dtsig$, we have
$$
{H_{g_0}(\Phi_T(\rho),\Phi_T(\gamma))\over H_{g_0}(\rho,\gamma)}
={1-|\bw|^2\over4|\by|^2}\cdot{4|T\by|^2\over1-|T\bw|^2}
={1-|\bw|^2\over1-|T\bw|^2}\cdot{|T\by|^2\over|\by|^2}
$$
so that
$$
\eta_{g_0}^\relent(\Phi_T)
=\sup_{|\bw|<1}{1-|\bw|^2\over1-|T\bw|^2}\cdot\sup_{\by\ne0}{|T\by|^2\over|\by|^2}
=\norm{T}_\infty^2.
$$

Next, for $s>0$ we use Lemma~\ref{lemm:s0} with $\by\ne0$, $\bu = \bw - s \bx$ and
$\bv = \bw + s \bx$. Note that since $\by = \bw-\bx = \frac{1}{2s} [(1+s) \bu + (1-s) \bv]$,
$\by$ is orthogonal to $\bu \times \bv$ so that we can use the reduction by
Lemma~\ref{lemm2:s0} to conclude
$$
\eta_{g_{s}}^{\RE}(\Phi_{T}) = \sup
\frac {\bigl\bra T \by, \left( \left[(1+s)^2 - |T\bu|^2 \right] I + |T\bu \ket \bra T\bu|
- |T\bv \ket \bra T\bv | \right) ^{-1} T\by \bigr\ket}
{\bra \by, \left( \left[(1+s)^2 - |\bu|^2 \right] I + |\bu \ket \bra \bu |
- | \bv \ket \bra \bv | \right)^{-1} \by \ket},
$$
where the supremum is taken over
\begin{eqnarray*} 
\lefteqn{ \{ (\bu, \bv, \by) : \by = \bw - \bx,\,\bu = \bw - s \bx,\,\bv = \bw + s \bx,
\,|\bw| < 1,\,|\bx| < 1 \} } \\
& = & \biggl\{ (\bu, \bv, \by) : \by = \frac{1}{2s}[(1+s)\bv - (1-s)\bu],\,|\bu + \bv| < 2,
\,|\bu - \bv| < 2s \biggr\}.
\end{eqnarray*}
Since relaxing the constraints can only increase the supremum, we will let
$\by\,(\ne0) \in {\bf R}^3$ arbitrary, and since
$$
\{ (\bu, \bv) : |\bu + \bv| < 2,\, |\bu - \bv| < 2s \}  \subset 
{\cal M}_{s} \equiv \{ (\bu, \bv) :  |\bu| < 1 + s,\, |\bv| < 1 + s \},
$$
we allow $(\bu, \bv) \in {\cal M}_s$. Then we find using Lemma~\ref{lemm:rev} that
\be \label{eta:ext}
\eta_{g_{s}}^{\RE}(\Phi_{T}) & \leq & \sup_{{\cal M}_{s}} \sup_{\by\ne0} 
\frac { \bigl\bra T^* \by, \left( \left[(1+s)^2 - |\bu|^2 \right] I
+ |\bu \ket \bra \bu | - | \bv \ket \bra \bv | \right) T^* \by \bigr\ket }
{ \bigl\bra \by, \left( \left[(1+s)^2 - |T\bu|^2 \right] I
+ | T\bu \ket \bra T\bu | -  | T\bv \ket \bra T\bv | \right) \by \bigr\ket} \nn \\
& = & \sup_{{\cal M}_{s}} \sup_{\by\ne0}  
\frac{ |T^*\by|^2 [(1+s)^2 - |\bu|^2 ] + |\bra \by, T\bu \ket |^2 - |\bra \by, T\bv \ket |^2 }
{ |\by|^2 [(1+s)^2 - |T \bu|^2 ] + |\bra \by, T \bu \ket |^2 -  |\bra  \by, T \bv \ket |^2 }.
\ee  
Now consider  
\be  \label{compare}
\frac{ |T^*\by|^2 (1+s)^2 - |\bu|^2 |T^*\by|^2 + |\by|^2 |T \bu|^2 }{|\by|^2(1+s)^2}
= \frac{ | T^*\by |^2}{|\by|^2 } + \frac{|\bu|^2}{(1+s)^2}
\bigg[ \frac{ | T\bu |^2}{|\bu|^2 } - \frac{ | T^*\by |^2}{|\by|^2 } \bigg].
\ee
Now if the quantity in $[ ~~ ]$ on the RHS of \eqref{compare} is $ \leq 0$, then the LHS of 
\eqref{compare} is $\leq |T^*\by|^2/|\by|^2 \leq \norm{T^*}_\infty^2=\norm{T}_\infty^2\le 1$.
On the other hand, if it is $>0 $ then from $|\bu|\le 1+s$ the LHS of \eqref{compare} is
$\leq |T\bu|^2/|\bu|^2 \leq \norm{T}_\infty^2\le 1$. Thus, if we let
$a = |T^*\by|^2 (1+s)^2 - |\bu|^2 |T^*\by|^2 + |\by|^2 |T \bu|^2$, $b = |\by|^2 (1+s)^2 $,
and $c = |\by|^2 |T \bu|^2 - |\bra \by, T \bu \ket |^2 + |\bra \by, T \bv \ket |^2 $,
then the fraction in \eqref{eta:ext} is $(a-c)/(b-c)$. Since $a>c\ge0$ and
$a/b\le\norm{T}_\infty\le1$ by the analysis above, for every $\by\ne0$ and
$(\bu,\bv)\in{\cal M}_s$ we have
$$
{a-c\over b-c}\le{a\over b}\le\norm{T}_\infty^2,
$$
which implies that $\eta_{g_s}^\relent(\Phi_T)\le\norm{T}_\infty^2$. Therefore,
$\eta_g^\relent(\Phi_T)\le\norm{T}_\infty^2$ holds for all $g\in\cG$. \qed


\subsection{Non-unital CQ qubit channels: Proof of Theorem \ref{thm:etaCQ}}

Before proving Theorem~\ref{thm:etaCQ} observe that the effect of $\Phi\equiv\Phi_{\alpha,\tau}$
on $P = \half(I + \bw \dtsig)$ is to map $\bw\mapsto (\alpha w_1, 0, \tau)^t$ and on
$A = \by \dtsig$ it is $\by\mapsto (\alpha y_1, 0, 0)^t$. There is no loss of generality
in simply using $P = I + \bw \dtsig$ whenever the factors of $\half$ will cancel.

\subsubsection{Dobrushin coefficient}
As in the proof of \eqref{dob.unit}, 
$$
\eta^{\tr}(\Phi_{\alpha,\tau}) = \half \sup_{ |\bw| = 1} \tr \,|\alpha w_1 \sigma_1|= \alpha
$$
with the supremum attained at $\bw = (1,0,0)^t$.
 
\subsubsection{Maximal metric}
It follows from \eqref{matinv} and \eqref{matprod} that for $P=I+\bw\dtsig$,
$$
\bigl\<\by\dtsig,\Omega_P^{\max}(\by\dtsig)\bigr\>
= \Tr(\by\dtsig)^2P^{-1}= \frac{2|\by|^2}{1 - |\bw|^2 }
$$
so that 
\be \label{etamax.qnon}
\eta_{\max}^{\riem}(\Phi_{\alpha,\tau})  
& = & \alpha^2 \sup_{|\bw|<1} \sup_{\by\ne0} \frac{ y_1^2}{|\by|^2}
\frac{1 - |\bw|^2}{1 - \alpha^2 w_1^2 - \tau^2} \nn \\
& = & \alpha^2 \sup_{|w_1|<1} \frac{1 - w_1^2} {1 - \alpha^2 w_1^2 - \tau^2} \\
& = & \frac{\alpha^2}{1 - \tau^2}  ~ \leq ~ 1, \nn
\ee 
since it is easy to verify the last equality when $\alpha^2 + \tau^2 \leq 1$. 


\subsubsection{$\kappa(x) = x^{-1/2}$}   \label{pf6.2d}

For $\kappa(x) = x^{-1/2}$, first  observe that it follows from \eqref{matinv} and \eqref{matsqrt} that
\be  \label{sqrtinv}
(\by \dtsig) (I + \bw \dtsig)^{-1/2} = \sqrt{ \tfrac{\zeta}{2 (1- |\bw|^2)}} 
\Big[ \tfrac{-1}{\zeta} ( \bw \cdot \by) I
+ \big[\by + \tfrac{i}{\zeta} (\bw \times \by) \big] \dtsig \Big]
\ee
with $\zeta = \zeta(|\bw|)\equiv 1 + \sqrt{1 - |\bw|^2}$ so that for $P=I + \bw \dtsig$
\be   \label{theta.reduc}
\Tr (\by \dtsig) \Omega^{x^{-1/2}}_P (\by \dtsig)
& = & \Tr \Big[ (\by \dtsig) (I + \bw \dtsig)^{-1/2} )\Big]^2 \nn \\
& = & \frac{\zeta}{2(1 - |\bw|^2 )} \, \Big[ \tfrac{1}{\zeta^2}
(\bw \cdot \by)^2 + |\by|^2 - \tfrac{1}{\zeta^2} | \bw \times \by |^2 \Big] \qquad \nn \\
& = & \frac{\zeta}{2(1 - |\bw|^2 )} \, \Big[ \tfrac{1}{\zeta^2}
| \bw|^2 |\by|^2 ( \cos^2 \theta - \sin^2 \theta)  + | \by|^2 \Big] \nn \\
& = & \frac{ |\by|^2 }{ 2(1 - |\bw|^2) \zeta } \, \Big( | \bw|^2 \cos 2 \theta
+ \zeta^2 \Big),
\ee
where $\theta$ denotes the angle between $\bw$ and $\by$.  Thus  
\be  \label{etasqrtexact}
\lefteqn{ \eta_{x^{-1/2}}^{\riem}(\Phi_{\alpha,\tau}) } \\  \nn
& = & \sup_{\bw, \by} \frac{ \alpha^2 y_1^2}{|\by|^2} \,
\frac{1 - |\bw|^2 } {1 - \alpha^2 w_1^2 - \tau^2} \, 
\frac{\zeta(|\bw|)} {\zeta(|(\alpha w_1, 0, \tau) |)}
\frac { (\alpha^2 w_1^2 + \tau^2) \cos 2 \wtd{ \theta} + \zeta^2(|(\alpha w_1, 0, \tau) |)}
{| \bw|^2 \cos 2 \theta + \zeta^2(|\bw|)}
\ee
with $\wtd{\theta}$ the angle between  $ (\alpha y_1, 0 , 0)$ and $(\alpha w_1, 0, \tau) $.
The first ratio in this product is largest when $\by = (y_1, 0 , 0) $ and $\by$ enters only
implicitly in the last one in $\theta, \wtd{\theta}$. Assuming $\by = (y_1, 0 , 0)$ and
$w_1 = 0$ gives $\wtd{\theta} =\theta = \pi/2$ and $\cos 2 \wtd{\theta} = \cos 2 \theta = -1$.
Then the identity $\zeta^2(x) - x^2 = 2 \sqrt{1 - x^2}\,\zeta(x) $ can be used to simplify
the RHS of \eqref{etasqrtexact} to give
\bee  
\eta_{x^{-1/2}}^{\riem}(\Phi_{\alpha,\tau})
& \geq & \sup_{\bw = (0,w_2,w_3)} \alpha^2 \, \frac{1 - |\bw|^2 } {1 - \tau^2}\, 
\frac{\zeta(|\bw|) } {\zeta(|\tau| )}\,  
\frac{ 2 \sqrt{1 - \tau^2}\,\zeta(|\tau|) }{ 2 \sqrt{1 - |\bw|^2} \zeta(|\bw|) } \nonumber\\
& = & \sup_{\bw = (0,w_2,w_3)} \alpha^2 \frac{ \sqrt{ 1 - |\bw|^2}}{ \sqrt{1 - \tau^2} }
~  =  ~ \frac{ \alpha^2}{ \sqrt{1 - \tau^2} },
\eee
which is equivalent to \eqref{etageomopt}.

We could have obtained this bound more easily  by  considering the special case $P = I$.
However, since the methods used to obtain \eqref{theta.reduc} are used again later, there
is some merit in presenting the details in this relatively simple setting.  For the special  case,
observe that 
$\Phi(I) = I + \tau \sigma_3 $ so that
$ [ \Phi(I)]^{-1/2}  =\pmx  \tfrac{1}{\sqrt{1+ \tau} } & 0 \\ 0 & \tfrac{1}{\sqrt{1- \tau} } \emx$
and
$$
\Phi(\by \dtsig) [ \Phi(I)]^{-1/2} \Phi(\by \dtsig) [ \Phi(I)]^{-1/2}  = \alpha^2 y_1^2 
 \frac{1}{\sqrt{1- \tau^2} } I,
$$
from which it follows   that
\be  \label{etageomopt}
\eta_{x^{-1/2}}^{\riem}(\Phi_{\alpha,\tau}) \geq \sup_{\by\ne0} \alpha^2 \frac{ y_1^2}{|\by|^2}
\frac{1}{ \sqrt{1- \tau^2}} = \frac{\alpha^2}{ \sqrt{1-t^2}}.
\ee

\subsubsection{BKM metric}    \label{pf6.2e}

For $\kappa_\BKM(x) = (\log x)/(x - 1)$ we can repeat some of the strategy above to get a
lower bound. It follows from \eqref{Omega:BKM} that
\bee
\Omega_P^\BKM( A) = \int_0^{\infty} \frac{1}{P+ uI}\, A\, \frac{1}{P+ uI}  \,du.
\eee
Then using \eqref{matinv} with $a = 1 + u$ and \eqref{matprod}, we find for  $A = \by \dtsig $
and $P = I + \bw \dtsig $
$$
(\by \dtsig) \frac{1}{(1 + u) I + \bw \dtsig }
= \frac{ - \bw \cdot \by I  +[ (1+u) \by + i \bw \times \by] \dtsig }{(1+u)^2 - |\bw|^2},
$$
from which it follows that
\bee
\lefteqn{ \Tr (\by \dtsig) \frac{1}{(1 + u) I + \bw \dtsig} (\by \dtsig)
\frac{1}{(1 + u) I + \bw \dtsig }  
= \Tr\bigg[ (\by \dtsig) \frac{1}{(1 + u) I + \bw \dtsig }\bigg]^2 } \nn
\qquad \qquad  \\
& = & \frac{2}{\big[(1+u)^2 - |\bw|^2 \big] ^2}\,
\Big[ |\bw \cdot \by |^2 + (1+u)^2 |\by|^2 - | \bw \times \by|^2 \Big] \nn \\
& = & \frac{2\,|\by|^2}{\big[(1+u)^2 - |\bw|^2\big]^2}\,
\big[ (1+u)^2 + |\bw |^2 \cos 2 \theta \big],  \label{log3}
\eee
where for the second equality above we have used $\by\cdot(\bw\times\by)=0$ and $\theta$ is
the angle between $\bw$ and $\by$. For $\theta = \pi/2$,
$ \cos 2 \theta = -1$ so that \eqref{log3} becomes $ |\by|^2/[(1+u)^2 - |\bw|^2] $. Then
since the integral  
$$
\int_0^\infty \frac{1}{(1+u)^2 - |\bw|^2 }\,du = \frac{1}{2 |\bw| }
\log \tfrac{1 + | \bw |}{1 - | \bw |}
$$
is elementary, we can conclude
\be
\eta_\BKM^\riem(\Phi) & \geq & \sup_{ w_1 = 0,\, \by \cdot \bw = 0 }
\alpha^2 \frac{ y_1^2}{|\by|^2} \frac{|\bw|}{\tau}
\frac{\log \tfrac{ 1 + \tau }{ 1 -\tau}} { \log \tfrac{1 + |\bw|}{1 - |\bw|}}  \nn \\
& = & \sup_{\bw = (0, w_2, w_3) } \frac{ \alpha^2 }{\tau} \log \tfrac{ 1 + \tau}{ 1 -\tau}
\frac{|\bw|} { \log \tfrac{1 + |\bw|}{1 - |\bw|}} \nn \\
& = & \frac{ \alpha^2 }{2\tau} \log \tfrac{ 1 + \tau}{ 1 -\tau},
\ee
since the inequality $ \log \frac{1+x}{1-x} \geq  2x $ implies that  
$$
\sup_{|\bw| < 1 } \,  \frac{|\bw|} { \log \tfrac{1 + |\bw|}{1 - |\bw|}}
= \lim_{ |\bw| \raw 0}  \, \frac{|\bw|} { \log \tfrac{1 + |\bw|}{1 - |\bw|}} = {1\over2}.
$$
This appears to be a reasonable bound when $\tau $ is small. Although it might appear to blow up when
$\tau \raw 1$, the CP condition $\alpha^2 + \tau^2 = 1$ implies that if $\tau \raw 1$ then
$\alpha \raw 0 $.
Since $ 1 < \frac{1}{2\tau} \log \tfrac{ 1 + \tau}{ 1 -\tau} < {1\over 1-\tau^2} $
 for all $\tau\in(0,1)$,  we can conclude that  $\alpha^2<\eta_\BKM^\riem(\Phi)<  \frac{ \alpha^2 }{1-\tau^2}$
consistent with  \eqref{qubitCQgen}.%
%

\subsubsection{Dual of WY metric}
For the WY function $\kappa_\WY(x) = 4/(1+ \sqrt{x})^2$, the dual function
$1/\kappa_\WY(x^{-1}) = (1+ \sqrt{x})^2/4x $ gives the operator 
$$
\Omega_P^{\wh\WY} = \tfrac{1}{4} \big( L_P^{-1/2} + R_P^{-1/2} \big)^2.
$$
To proceed as above, for $P=I+\bw\dtsig$ we first find
\be
( L_P^{-1/2} + R_P^{-1/2} \big)(\by \dtsig )
& = & (I + \bw \dtsig)^{-1/2} (\by \dtsig ) + (\by \dtsig ) (I + \bw \dtsig)^{-1/2} \nn \\
& = & \sqrt{\frac{2\zeta(|\bw|)}{ 1- |\bw|^2}}\, 
\Big[ \tfrac{-1}{\zeta(|\bw|)}( \by \cdot \bw) I + \by \Big],
\ee
where we have used \eqref{sqrtinv} and the fact that the $\bw \times \by$ terms have the
opposite sign and cancel.    As in Sections \ref{pf6.2d} and \ref{pf6.2e} the resulting expressions
are difficult to deal with unless we make the simplifying assumption $\by = (y_1, 0, 0)^t$ and
$w_1 = 0 $  which yields the lower bound
%
\bee
\eta_{\wh{WY}}^\riem(\Phi)
&\geq & \alpha^2\, \frac{ 1 + \sqrt{1 -  \tau^2}} { 1 - \tau^2}
\sup_{\bw = (0, w_2, w_3)} \frac{1 - |\bw|^2 }{1 + \sqrt{ 1 - | \bw|^2}} \nn \\
& = & \alpha^2\, \frac{ 1 + \sqrt{1 - \tau^2}} {2( 1 - \tau^2)}.
\eee


\subsubsection{Minimal or Bures metric} 
For the smallest function $\kmin(x) = 2/(1+x)$ in $\cK$ we begin by using \eqref{prebur} of
Lemma~\ref{lemm:bures} to conclude that for $P = I + \bw \dtsig$
$$
\Tr (\by \dtsig) \frac{2}{L_P+ R_P} (\by \dtsig)
= \by\cdot\bigg[ \by + \frac{ (\bw \cdot \by) \, \bw }{1 - |\bw|^2} \bigg]
= \frac{|\by|^2(1-|\bw|^2)+(\bw\cdot\by)^2}{1-|\bw|^2}.
$$
Thus
$$
\eta_{\min}^\riem(\Phi) = \sup_{|\bw| < 1} \sup_{ \by \ne 0 } 
\frac{ \alpha^2 y_1^2(1 - |\bw|^2) }{ |\by|^2(1 - |\bw|^2 ) + (\bw \cdot \by)^2 }\,
\frac{ 1 - \tau^2 }{1 - \tau^2 - \alpha^2 w_1^2 }.
$$
For fixed $\bw$ we can optimize the first term, which depends only on the ratios $y_2/y_1$ and $y_3/y_1,$ directly or use
Lemma~\ref{lemm:Fmin} to conclude that the maximum is achieved when
$\by = \bigl(1, - w_1w_2/(1-w_1^2), -w_1w_3/(1-w_1^2)\bigr)^t$. In this case,
$$
{y_1^2(1-|\bw|^2)\over |\by|^2(1 - |\bw|^2 ) + (\bw \cdot \by)^2}
={(1-|\bw|^2)^2\over 1-w_1^2}.
$$
Thus
$$
\eta_{\min}^\riem(\Phi) = \alpha^2 ( 1  - \tau^2) \sup_{|w_1| < 1}    
\frac{1-w_1^2}{1 - \tau^2 - \alpha^2 w_1^2 }.
$$
When the CP condition $\alpha^2 + \tau^2 \leq 1$ holds, it is elementary to verify that
\be \label{elemw1}
\frac{1 -w_1^2}{ 1- \tau^2 - \alpha^2 w_1^2 } \leq \frac{1}{1-\tau^2} 
\ee
 so that the supremum is achieved when
$w_1 = 0 $, and 
\be  \label{etamin.qnon}
\eta_{\min}^\riem(\Phi) = \alpha^2.
\ee


\subsubsection{Extreme points}

The extreme functions have the form
$$
\kappa_s(x) = \frac{1+s}{2} \bigg( \frac{1}{x +s } + \frac{1}{1 +sx } \bigg),
$$
for which the corresponding operator is
$$
\Omega_P^{\kappa_s}
= \frac{1+s}{2} \bigg( \frac{1}{L_P + s R_P} + \frac{1}{R_P + s L_P} \bigg).
$$
Observe that for any function $h:(0, \infty) \to {\bf R}$ and $X\in\bM_d$, 
$$
\big[ R_P^{-1} h\big(L_P R_P^{-1} \big) (X) \big]^*
= L_P^{-1} h\big( R_P L_P^{-1} \big) (X^*)
= R_P^{-1} \wtd{h}\big(L_P R_P^{-1} \big) (X^*),
$$
where $\wtd{h}(x)\equiv x^{-1}h(x^{-1})$. Using this for the function $h(x)= (1+s)/(x +s )$
for which $\wtd{h}(x) = (1+s)/(1 +sx )$, we have for $A = A^*$
\be  \label{eq:Psym}
\bra A, \Omega_P^{\kappa_s}(A) \ket = \bigg\bra A, \frac{1+s}{L_P + s R_P} (A) \bigg\ket
= \bigg\bra A, \frac{1+s}{R_P + s L_P} (A) \bigg\ket.
\ee
Now we let $P = I + \bw \dtsig$, $A = \by \dtsig$ and apply Lemma~\ref{lemm:s0} 
with $\bu = (1-s) \bw$, $\bv = (1+s) \bw$ to obtain
\bee
\bra A, \Omega_P^{\kappa_s} (A) \ket
& = & 2(1+s)^2 \Big\bra \by, \big[\xi_s(|\bw|^2) I
- 4s \proj{\bw} \big]^{-1} \by \Big\ket \nn\\
& = & \frac{2(1+s)^2}{\xi_{s}(|\bw|^2)} \bigg[ |\by|^2
+ \frac{ 4 s (\bw \cdot \by )^2}{(1+s)^2 (1 - |\bw|^2)} \bigg],
\eee
where we have used \eqref{eq:inv.3dim} and 
\be  \label{xidef}
\xi_s(x) \equiv (1+s)^2 - (1-s)^2 x = (1+s)^2 (1-x) + 4sx.
\ee
Thus
\begin{align*}
& \frac{ \bra \Phi(A), \Omega_{\Phi(P)}^{\kappa_s} \Phi(A) \ket }
{ \bra A, \Omega_P^{s} A \ket} \\
&\quad = \frac{ (1 - |\bw|^2 ) \xi_s(| \bw |^2)}
{|\by|^2(1+s)^2(1-|\bw|^2) + 4 s (\bw \cdot \by )^2} \,
\frac{\alpha^2 y_1^2 \big[(1+s)^2(1-\alpha^2 w_1^2 - \tau^2) + 4 s\alpha^2 w_1^2 \big]}
{(1-\alpha^2 w_1^2 - \tau^2)\,\xi_{s}(\alpha^2 w_1^2 + \tau^2)}
\end{align*}
so that
\begin{align}\label{eq:eta.non-un.a}
& \eta_{\kappa_{s}}^{\riem}(\Phi) \\  \nn
&\quad = \alpha^2 \sup_{|\bw|<1} \Bigg[\frac {(1 - |\bw|^2) \xi_{s}(|\bw|^2) } 
{(1-\alpha^2 w_1^2 - \tau^2)\,\xi_{s}(\alpha^2 w_1^2 + \tau^2)} \, 
\sup_{\by\ne0} \frac{ y_1^2 \big[(1+s)^2(1-\alpha^2 w_1^2 - \tau^2) + 4 s\alpha^2 w_1^2 \big] }
{|\by|^2(1+s)^2(1-|\bw|^2) + 4 s (\bw \cdot \by )^2} \Bigg].
\end{align}
We first consider $\sup_\by$ with $\bw$ fixed.
This term depends only on the ratios $y_2/y_1$, $y_3/y_1$ so that there is no loss of
generality in assuming $y_1 = 1$, in which case only the denominator depends on $\by$ and
we consider instead its minimum, i.e., we seek
\begin{eqnarray*}
\min_{y_2,\,y_3} \left[
(1 + y_2^2 + y_3^2) (1+s)^2(1-|\bw|^2) + 4 s (w_1 + w_2 y_2 + y_3 y_3 )^2 \right]
\end{eqnarray*}
which is found in Lemma~\ref{lemm:Fmin} with $\mu = (1+s)^2(1-|\bw|^2)$ and $\nu = 4s$.
The minimum is  
\be  \label{Fminext}
\frac{ (1+s)^2 (1 - |\bw|^2) \xi_s(|\bw|^2) } { \xi_s(|\bw|^2) - 4 s w_1^2 }.
\ee
Inserting \eqref{Fminext} into \eqref{eq:eta.non-un.a} yields
\begin{eqnarray*}
\eta_{\kappa_{s}}^{\riem}(\Phi) = \alpha^2 \sup_{|\bw|<1}
\frac { \bigl[(1+s)^2(1-\alpha^2 w_1^2 - \tau^2) + 4 s \alpha^2 w_1^2 \bigr]
\,\bigl[\xi_{s}(|\bw|^2) - 4 s w_1^2\bigr]} 
{(1-\alpha^2 w_1^2 - \tau^2)\,\xi_{s}(\alpha^2 w_1^2 + \tau^2)\, (1+s)^2}.
\end{eqnarray*}
The only term in the expression above which involves $\bw$ rather than $w_1$ is
$\xi_{s}(|\bw|^2) \linebreak = (1+s)^2 - (1-s)^2|\bw|^2$, which is largest when $|\bw|$ is smallest,
i.e., $|\bw| = |w_1|$ or, equivalently, $\bw = (w_1,0,0)$.
Thus we find 
\begin{eqnarray}\label{eq:eta.non-un.d}
\eta_{\kappa_{s}}^{\riem}(\Phi) = \alpha^2 \sup_{|w_1|<1} 
\frac { (1+s)^2(1-\alpha^2 w_1^2 - \tau^2) + 4s \alpha^2 w_1^2 }
{(1+s)^2(1-\alpha^2 w_1^2 - \tau^2) + 4s (\alpha^2 w_1^2 + \tau^2)} \,
\frac{1 - w_1^2} {1-\alpha^2 w_1^2 - \tau^2}.
\end{eqnarray}
For $s = 0, 1$ this reduces to the expression for the maximal and minimal functions,
\eqref{etamax.qnon} and \eqref{elemw1}, respectively. As in \eqref{elemw1} the second factor
is largest when $w_1 = 0$. The first factor can be written as
\bee
\frac{R }{R + 4s \tau^2} = 1 - \frac{ 4s \tau^2}{R + 4s \tau^2},
\eee
which is largest when 
$$
R \equiv (1+s)^2(1-\alpha^2 w_1^2 - \tau^2) + 4s \alpha^2 w_1^2
= (1+s)^2 (1 - \tau^2) - (1-s)^2 \alpha^2 w_1^2
$$
is largest, which also occurs when $w_1^2 = 0 $. Using these observations in
\eqref{eq:eta.non-un.d} we can conclude that
\begin{eqnarray}\label{eq:eta.non-un.e}
\eta_{\kappa_{s}}^{\riem}(\Phi) = \frac{ (1+s)^2 \alpha^2}{(1+s)^2 -(1-s)^2 \tau^2 }
= \frac{ \alpha^2}{1 - \left(\frac{1-s}{1+s}\right)^2 \tau^2 }.  
\end{eqnarray}
When $s = 0, 1$ we recover the expressions \eqref{etamax.qnon} and \eqref{etamin.qnon}.

\subsubsection{Wigner-Yanase metric} 

Although we have had to make simplifying assumptions to obtain lower
bounds for all but the extremal $\kappa$, it is 
quite remarkable that we can obtain an exact expression in the case
of the Wigner-Yanase function $\kappa_\WY(x)=4/(1+\sqrt x)^2$.  Then
$$
\Omega_P^\WY={4\over\bigl(\sqrt{L_P}+\sqrt{R_P}\bigr)^2}.
$$
For $P=I+\bw\dtsig$, using \eqref{matsqrt} and Lemma \ref{lemm:bures} we can write
$$
{2\over\sqrt{L_P}+\sqrt{R_P}}(\by\dtsig)
=\sqrt{\frac{2}{\zeta}}\bigg[-{\zeta\bw\cdot\by\over\zeta^2-|\bw|^2}I
+\biggl(\by+{(\bw\cdot\by)\bw\over\zeta^2-|\bw|^2}\biggr)\dtsig\bigg]
$$
with $\zeta=\zeta(|\bw|)=1+\sqrt{1-|\bw|^2}$. Therefore,
\bee
  \tr (\by\dtsig){4\over\bigl(\sqrt{L_P}+\sqrt{R_P}\bigr)^2}(\by\dtsig)
 & = & \tr \biggl[{2\over\sqrt{L_P}+\sqrt{R_P}}(\by\dtsig)\biggr]^2 \nn\\
& = & \frac{4}{\zeta}\bigg[{\zeta^2(\bw\cdot\by)^2\over(\zeta^2-|\bw|^2)^2}+|\by|^2
+2\,{(\bw\cdot\by)^2\over\zeta^2-|\bw|^2}
+{(\bw\cdot\by)^2|\bw|^2\over(\zeta^2-|\bw|^2)^2}\bigg] \nn\\
& = & \frac{4}{\zeta}\bigg[|\by|^2
+(\bw\cdot\by)^2\,{3\zeta^2-|\bw|^2\over(\zeta^2-|\bw|^2)^2}\bigg] \label{WY-metric}   \\
%
& = & 4\,\big\<\by,\big[\zeta I-(2-\zeta^{-1})|\bw\>\<\bw|\big]^{-1}\by\big\>,
 \eee
where the last equality is the key to our ability to evaluate  $\eta_\WY^\riem(\Phi)$ exactly.
To obtain this, one can apply $|\bw|^2=\zeta(2-\zeta)$ and \eqref{eq:inv.3dim} to see that
\begin{align*}
\frac{1}{\zeta}\bigg[I+{3\zeta^2-|\bw|^2\over(\zeta^2-|\bw|^2)^2}\,|\bw\>\<\bw|\bigg]
&=\frac{1}{\zeta}\bigg[I+{2\zeta-1\over2\zeta(\zeta-1)^2}\,|\bw\>\<\bw|\bigg] \\
&=\big[\zeta I-(2-\zeta^{-1})|\bw\>\<\bw|\big]^{-1}.
\end{align*}
Then with  $T\equiv\diag(\alpha,0,0)$, $\widetilde\bw\equiv(\alpha w_1,0,\tau)^t$ and
$\widetilde\zeta\equiv\zeta(|\widetilde\bw|)$,
we use Lemma \ref{lemm:rev} to obtain
\begin{align*}
\eta_\WY^\riem(\Phi)
&=\sup_{|\bw|<1}\sup_{\by\ne0}
{\big\<T\by,\big[\widetilde\zeta I-(2-\widetilde\zeta^{-1})
|\widetilde\bw\>\<\widetilde\bw|\big]^{-1}T\by\big\>
\over\big\<\by,\big[\zeta I-(2-\zeta^{-1})|\bw\>\<\bw|\big]^{-1}\by\big\>} \\
&=\sup_{|\bw|<1}\sup_{\by\ne0}
{\big\<T\by,\big[\zeta I-(2-\zeta^{-1})|\bw\>\<\bw|\big]T\by\big\>
\over\big\<\by,\big[\widetilde\zeta I-(2-\widetilde\zeta^{-1})
|\widetilde\bw\>\<\widetilde\bw|\big]\by\big\>} \\
&=\sup_{|\bw|<1}\sup_{\by\ne0}
{\zeta\,|T\by|^2-(2-\zeta^{-1})|\<\bw,T\by\>|^2
\over\widetilde\zeta\,|\by|^2-(2-\widetilde\zeta^{-1})|\<\widetilde\bw,\by\>|^2} \\
&=\sup_{|\bw|<1}\sup_{(y_1,y_3)\ne(0,0)}
{\alpha^2\bigl[\zeta-w_1^2(2-\zeta^{-1})\bigr]y_1^2\over
\widetilde\zeta(y_1^2+y_3^2)-(2-\widetilde\zeta^{-1})(\alpha w_1y_1+\tau y_3)^2}.
\end{align*}
Since the denominator of the last ratio is
$$
\bigl[\widetilde\zeta-\tau^2(2-\widetilde\zeta^{-1})\bigr]
\Biggl(y_3-{\alpha\tau w_1(2-\widetilde\zeta^{-1})y_1\over
\widetilde\zeta-\tau^2(2-\widetilde\zeta^{-1})}\Biggr)^2
+{\widetilde\zeta^2-(\alpha^2w_1^2+\tau^2)(2\widetilde\zeta-1)\over
\widetilde\zeta-\tau^2(2-\widetilde\zeta^{-1})}\,y_1^2,
$$
it follows that
$$
\eta_\WY^\riem(\Phi)=\sup_{|\bw|<1}
{\alpha^2\bigl[\zeta-w_1^2(2-\zeta^{-1})\bigr]
\bigl[\widetilde\zeta-\tau^2(2-\widetilde\zeta^{-1})\bigr]
\over\widetilde\zeta^2-(\alpha^2w_1^2+\tau^2)(2\widetilde\zeta-1)}.
$$
Then $\widetilde\zeta=1+\sqrt{1-\tau^2-\alpha^2w_1^2}$ depends on $w_1$ only and
\begin{align*}
\zeta-w_1^2(2-\zeta^{-1})
&=1+\sqrt{1-|\bw|^2}-w_1^2\Biggl(2-{1-\sqrt{1-|\bw|^2}\over|\bw|^2}\Biggr) \\
&=2(1-w_1^2)-{(w_2^2+w_3^2)\bigl(1-\sqrt{1-|\bw|^2}\bigr)\over|\bw|^2},
\end{align*}
which obviously takes the maximum $2(1-w_1^2)$ when $w_2=w_3=0$. Therefore,
\begin{align*}
\eta_\WY^\riem(\Phi_{\alpha,\tau})
&=\sup_{|w_1|<1}
{2\alpha^2(1-w_1^2)\Bigl[1+\sqrt{1-\tau^2-\alpha^2w_1^2}
-\tau^2\Bigl(2-{1\over1+\sqrt{1-\tau^2-\alpha^2w_1^2}}\Bigr)\Bigr]\over
\bigl(1+\sqrt{1-\tau^2-\alpha^2w_1^2}\bigr)^2
-(\alpha^2w_1^2+\tau^2)\bigl(1+2\sqrt{1-\tau^2-\alpha^2w_1^2}\bigr)} \nonumber\\
&=\sup_{|w_1|<1}
{\alpha^2(1-w_1^2)\bigl[2(1-\tau^2)-\alpha^2w_1^2
+2(1-\tau^2)\sqrt{1-\tau^2-\alpha^2w_1^2}\bigr]\over
(1-\tau^2-\alpha^2w_1^2)\bigl(1+\sqrt{1-\tau^2-\alpha^2w_1^2}\bigr)^2} \\
&\le\alpha^2\left[\sup_{|w_1|<1}{1-w_1^2\over1-\tau^2-\alpha^2w_1^2}\right]
\left[\sup_{|w_1|<1}{2(1-\tau^2)\bigl(1+\sqrt{1-\tau^2-\alpha^2w_1^2}\bigr)
-\alpha^2w_1^2\over\bigl(1+\sqrt{1-\tau^2-\alpha^2w_1^2}\bigr)^2}\right].
\end{align*}
 As in \eqref{elemw1} the first supremum is attained when $w_1=0$.
For the second, let $\rho\equiv1-\tau^2$ and $x\equiv\alpha^2w_1^2\in[0,\alpha^2)$ where
$\alpha^2\le\rho<1$, and observe that the ratio can be written as
$$
{2\rho(1+\sqrt{\rho-x})-x\over(1+\sqrt{\rho-x})^2}
=\rho+(1-\rho){\rho-x\over(1+\sqrt{\rho-x})^2},
$$
which is maximized when $x=0$ (i.e., $w_1=0$). Thus we conclude
$$
\sup_{|w_1|<1}{2(1-\tau^2)\bigl(1+\sqrt{1-\tau^2-\alpha^2w_1^2}\bigr)
-\alpha^2w_1^2\over\bigl(1+\sqrt{1-\tau^2-\alpha^2w_1^2}\bigr)^2}
=\rho+(1-\rho){\rho\over(1+\sqrt\rho)^2}={2(1-\tau^2)\over1+\sqrt{1-\tau^2}}.
$$
Combining the two suprema yields
$$
\eta_\WY^\riem(\Phi_{\alpha,\tau})
\le{2\alpha^2\over1+\sqrt{1-\tau^2}},
$$
It is straightforward to see that the special case  $\by = (y_1,0,0)^t $ and $\bw = 0$ yields
the reverse inequality.  

The proof of Theorem \ref{thm:etaCQ} is now complete. \qed


\subsection{Proof of Theorem~\ref{thm:relent.neq.riem}}

For $P=I+\bw\dtsig$, $Q=I+\bx\dtsig$ and $0\le s\le1$, by \eqref{H-gs} and Lemma \ref{lemm2:s0}
we have
\begin{align*}
{H_{g_s}(Q,P)\over2(1+s)}
&={1\over2(1+s)}\,\Tr(\by\dtsig){1\over L_P+sR_Q}(\by\dtsig) \\
&=\Bigl\<\by,\bigl[\bigl\{(1+s)^2-|\bu|^2\bigr\}I
+\proj{\bu}-\proj{\bv}\bigr]^{-1}\by\Bigr\>,
\end{align*}
where $\by=\bw-\bx$, $\bu=\bw-s\bx$ and $\bv=\bw+s\bx$, and we note that $\by$ is orthogonal
to $\bu\times\bv$. The formula for $H_{g_s}(P,Q)$ is similar with $\bw$, $\bx$ interchanged.
Since the first inequality of \eqref{relent-ext} holds in general by \eqref{eta-g-sym}, we will
estimate $\eta_{(g_s)_\sym}^\relent(\Phi)$ for $\Phi=\Phi_{\alpha,\tau}$. For this we take
$\bw=(w_1,0,0)^t$ and $\bx=0$ for simplicity, for which we have
\be \label{H-gs(QP)}
{H_{g_s}(Q,P)\over2(1+s)}={w_1^2\over(1+s)^2-w_1^2},\qquad
{H_{g_s}(P,Q)\over2(1+s)}={w_1^2\over(1+s)^2-s^2w_1^2}.
\ee
Since $\Phi(P)=I+\wtd\bw\dtsig$ and $\Phi(Q)=I+\wtd\bx\dtsig$ where
$\wtd\bw=(\alpha w_1,0,\tau)^t$ and $\wtd\bx=(0,0,\tau)^t$, we have the expression
\be \label{H-gs-Phi(QP)}
{H_{g_s}(\Phi(Q),\Phi(P))\over2(1+s)}
=\Bigl\<\wtd\by,\bigl[\bigl\{(1+s)^2-|\wtd\bu|^2\bigr\}I
+\proj{\wtd\bu}-\proj{\wtd\bv}\bigr]^{-1}\wtd\by\Bigr\>,
\ee
where $\wtd\by=\wtd\bw-\wtd\bx=(\alpha w_1,0,0)^t$,
$\wtd\bu=\wtd\bw-s\wtd\bx=(\alpha w_1,0,(1-s)\tau)^t$ and
$\wtd\bv=\wtd\bw+s\wtd\bx=(\alpha w_1,0,(1+s)\tau)^t$. The matrix form of the operator inside
$[\quad]^{-1}$ of \eqref{H-gs-Phi(QP)} is
$$
\begin{pmatrix}
\xi_s(\tau^2)-\alpha^2w_1^2 & 0 & -2s\alpha\tau w_1 \\
0 & \xi_s(\tau^2)-\alpha^2w_1^2 & 0 \\
-2s\alpha\tau w_1 & 0 & \xi_s(\tau^2)-4s\tau^2-\alpha^2w_1^2
\end{pmatrix},
$$
where $\xi_s(\cdot)$ is in \eqref{xidef}. The $(1,1)$-entry of the inverse of this matrix is
$\bigl[\xi_s(\tau^2)-4s\tau^2-\alpha^2w_1^2\bigr]/\det$ where $\det$ is the determinant of the
$2\times2$ matrix of the first and the third rows and columns. Therefore, the exact form of
\eqref{H-gs-Phi(QP)} is
\be \label{H-gs-Phi(QP)2}
{H_{g_s}(\Phi(Q),\Phi(P))\over2(1+s)}
={\alpha^2w_1^2\bigl[\xi_s(\tau^2)-\alpha^2w_1^2-4s\tau^2\bigr]\over
\bigl[\xi_s(\tau^2)-\alpha^2w_1^2\bigr]\bigl[\xi_s(\tau^2)-\alpha^2w_1^2-4s\tau^2\bigr]
-4s^2\alpha^2\tau^2w_1^2}.
\ee
A similar computation with $\wtd\bw,\wtd\bx$ interchanged yields
\be \label{H-gs-Phi(PQ)2}
{H_{g_s}(\Phi(P),\Phi(Q))\over2(1+s)}
={\alpha^2w_1^2\bigl[\xi_s(\tau^2)-s^2\alpha^2w_1^2-4s\tau^2\bigr]\over
\bigl[\xi_s(\tau^2)-s^2\alpha^2w_1^2\bigr]\bigl[\xi_s(\tau^2)-s^2\alpha^2w_1^2-4s\tau^2\bigr]
-4s^2\alpha^2\tau^2w_1^2}.
\ee

We define
\begin{align*}
H(s)&\equiv\lim_{|w_1|\nearrow1}{H_{g_s}(Q,P)+H_{g_s}(P,Q)\over2(1+s)}, \\
\wtd H(s)&\equiv
\lim_{|w_1|\nearrow1}{H_{g_s}(\Phi(Q),\phi(P))+H_{g_s}(\Phi(P),\Phi(Q))\over2(1+s)}.
\end{align*}
Since $\eta_{(g_s)_\sym}^\relent(\Phi)\ge\wtd H(s)/H(s)$ for every $s\in[0,1]$, we may compare
$\wtd H(s)/H(s)$ with $\eta_{\kappa_s}^\riem(\Phi)$ for $s$ near $1$. For this we observe by
\eqref{H-gs(QP)}, \eqref{H-gs-Phi(QP)2} and \eqref{H-gs-Phi(PQ)2} that $H(1)=2/3$ and
$$
\wtd H(1)={2\alpha^2(4-\alpha^2-4\tau^2)\over
(4-\alpha^2)(4-\alpha^2-4\tau^2)-4\alpha^2\tau^2}
={2\alpha^2(4-\alpha^2-4\tau^2)\over(4-\alpha^2)^2-16\tau^2}.
$$
Now assume that $4\tau^2>(1-\alpha^2)(4-\alpha^2)$ (in particular, $\alpha^2>0$). Since
$(4-\alpha^2)^2-16\tau^2>0$ (thanks to $\alpha^2+\tau^2\le1$ and $\tau^2<1$) and
$$
3(4-\alpha^2-4\tau^2)-\bigl[(4-\alpha^2)^2-16\tau^2\bigr]
=4\tau^2-(1-\alpha^2)(4-\alpha^2)>0,
$$
it follows that $\wtd H(1)/H(1)>\alpha^2$. From the continuity of the $s$-dependence of
$\wtd H(s)/H(s)$ and $\eta_{\kappa_s}^\riem(\Phi)$ in \eqref{eq:eta.non-un.e}, we arrive at
the first assertion stated in the theorem.

For the second assertion, when $\alpha^2+\tau^2=1$, a tedious computation gives
$$
{\wtd H(s)\over H(s)}=\alpha^2\,
{s(s+2)(2s+1)\bigl[12s(s+1)^2+(2s^4+s^3+s+2)\alpha^2\bigr]\over
(s^2+4s+1)\bigl[4s(s+1)+s^3\alpha^2\bigr]\bigl[4s(s+1)+\alpha^2\bigr]}.
$$
The limit of $\bigl[\wtd H(s)/H(s)\bigr]/\eta_{\kappa_s}^\riem(\Phi)$ as $\alpha^2=1-\tau^2$
tends to $0$ is
$$
{3s(s+2)(2s+1)\over(s^2+4s+1)(s+1)^2}.
$$
The numerator minus the denominator of the above ratio is $-(s^4-5s^2+1)$, which is positive
when $s^2>\tfrac{5-\sqrt{21}}{2}$. This yields the second assertion of the theorem. \qed

\subsection{Useful results}

\subsubsection{Basic formulas}
We observe that any Hermitian matrix with $\Tr A = 0$ can be written as $A = \by \dtsig$
with $\by \in {\bf R}^3,$ and that $\Tr (aI + {\bf w \dtsig}) = 2a.$

The following formulas will be useful:
\begin{align}
(aI + {\bw \dtsig})(bI + {\by \dtsig})
& = (ab + \bw \cdot \by)I + (a \by + b \bw + i \bw \times \by) \dtsig, \label{matprod}\\
(aI + {\bw \dtsig})^{-1}
& = \frac{1}{a^2 - |\bw|^2} (aI - \bw \dtsig), \label{matinv}\\      
(bI + \bw \dtsig)^{1/2} & = \sqrt{ \frac{\zeta(b,\bw)}{2}}
\bigg[ I + \frac{ \bw \dtsig}{\zeta(b,\bw)} \bigg], \label{matsqrt}
\end{align}
where $\zeta(b,\bw) \equiv b + \sqrt{b^2 - |\bw|^2}$.

It will be convenient to use the physicists bra and ket notation for vectors in ${\bf R}_3$
as well as ${\bf C}^d$ in which  $| \bx \ket \bra \bx|$ denotes $|\bx|^2$ times the projection
onto $\bx$, more generally $| \bw \ket \bra \bx| : \by \mapsto (\bx\cdot\by)\bw $. 
In that notation, if $a\ne0,b|\bw|^2$ then
\begin{eqnarray}\label{eq:inv.3dim}
\big( aI - b |\bw \ket \bra \bw | \big)^{-1}
= a^{-1} \bigg[ I + \frac{b}{a-b|\bw|^2}\, |\bw \ket \bra \bw | \bigg].
\end{eqnarray}

The following lemmas are useful in Sections B.1--B.3 to prove the theorems of Section 6.

\begin{lemma}  \label{lemm:rev}
Let $A, B $ be positive linear operators on ${\bf R}^3$.  Then
\begin{eqnarray*}
 \sup_{\by \ne0} \frac
  {\bra T \by, A^{-1} T \by \ket}
  {\bra \by, B^{-1} \by \ket}  =
 \sup_{\by \ne0} \frac
  {\bra T^* \by, B T^* \by \ket}
  {\bra \by, A\by \ket}. 
\end{eqnarray*}
\end{lemma}

\proof  Writing $\bz = B^{-1/2} \by$ one finds
\begin{eqnarray*}
 \sup_{\by \ne0} \frac {\bra T \by , A^{-1} T \by \ket}{\bra \by ,B^{-1} \by \ket} & = &
 \sup_{\bz \ne0} \frac{\bra B^{1/2} \bz, T^*A^{-1} T B^{1/2} \bz \ket}
  {\bra \bz , \bz \ket} \\
  & = & \| B^{1/2} T^*A^{-1} T B^{1/2} \| \\
  & = & \|A^{-1/2} T B T^* A^{-1/2} \| \\
  & = &  \sup_{\bz \ne0} \frac {\bra A^{-1/2} \bz, T B T^* A^{-1/2} \bz \ket}
  {\bra \bz , \bz \ket},  
\end{eqnarray*}
where we have used the fact that $\| \Upsilon^* \Upsilon \| = \| \Upsilon \Upsilon^* \|$ with
$\Upsilon =A^{-1/2} T B^{1/2}$. Then redefining $\by =A^{-1/2} \bz$ gives the desired result.
\qed

\begin{lemma}  \label{lemm:Fmin}
For any fixed $\mu>0$, $\nu\ge0$ and $\bw$, the minimum of
$$
F(y_2,y_3) \equiv \mu (1  + y_2^2 + y_3^2) + \nu (w_1 + w_2 y_2 + w_3 y_3)^2
$$
is 
$$
\min_{y_2,y_2}F(y_2,y_3) = \frac{ \mu(\mu + \nu \, |\bw|^2) }{\mu + \nu(|\bw|^2-w_1^2)}.
$$
\end{lemma}

\proof
The condition that $\nabla F =0 $ yields two linear equations which can be written in the form
$$
\pmx \mu + \nu \, w_2^2 & \nu \, w_2 w_3 \\ \nu \, w_2 w_3 & \mu + \nu \, w_3^2 \emx
\pmx y_2 \\y_3 \emx = - \nu \, w_1 \pmx w_2 \\ w_3 \emx.     
$$
This has the solution
\bee
\pmx y_2 \\y_3 \emx & = & \frac{-\nu \, w_1}{\det}
\pmx \mu + \nu \, w_3^2 & - \nu \, w_2 w_3 \\ - \nu \, w_2 w_3 & \mu + \nu \, w_2^2 \emx
\pmx w_2 \\ w_3 \emx  \nn \\
& = & \frac{- \mu \nu \, w_1}{\det} \pmx w_2 \\ w_3 \emx
~ = ~ \frac{- \nu \, w_1}{ \mu + \nu \, (w_2^2 + w_3^2) } \pmx w_2 \\ w_3 \emx
\eee
since $\det = \mu^2 + \mu \nu \, (w_2^2 + w_3^2)$. It is now easy to compute the value of $F$
at this solution as
\begin{align*}
&\frac{\mu\bigl[\mu+\nu(w_2^2+w_3^2)\bigr]^2+\mu\nu^2w_1^2(w_2^2+w_3^2)+\mu^2\nu w_1^2}
{\bigl[\mu+\nu(w_2^2+w_3^2)\bigr]^2} \\
&\qquad=\frac{\mu\bigl[\mu+\nu(w_2^2+w_3^2)\bigr]
+\mu\nu w_1^2}{\mu+\nu(w_2^2+w_3^2)}
=\frac{ \mu(\mu + \nu \, |\bw|^2) }{\mu + \nu(|\bw|^2-w_1^2)}. \qed
\end{align*}


\begin{lemma}  \label{lemm:bures}
Let $P = I + \bw \dtsig $ with  $| \bw | < 1$. Then for every $\by\in\bR^3$,
\be  \label{prebur}
\frac{2}{L_P+ R_P} (\by \dtsig) =   
 - \frac{ \bw \cdot \by }{1 - |\bw|^2} I
 + \bigg[ \by + \frac{ (\bw \cdot \by) \, \bw }{1 - |\bw|^2} \bigg] \dtsig
\ee
\end{lemma}

\proof
Write $\beta I+ \bz \dtsig=2(L_P+ R_P)^{-1}(\by \dtsig)$.
Then \eqref{matprod} and $ \bz \times \bw = - \bw \times \bz $ imply that
\bee
\by \dtsig & = & \half ( I + \bw \dtsig) ( \beta I+ \bz \dtsig )
+ \half ( \beta I+ \bz \dtsig ) ( I + \bw \dtsig) \nn \\
& = & (\beta + \bz \cdot \bw) I + ( \bz + \beta \bw) \dtsig.
\eee
Since $I$ and the Pauli matrices form a basis for $\bM_2$, this implies
$\beta = - \bz \cdot \bw$ and $\by=\bz+\beta\bw$ so that
$$
\by = \bz - ( \bz \cdot \bw ) \bw = (I - \proj{\bw} ) \bz.
$$
Then using \eqref{eq:inv.3dim} with $a = b =1$ we find that 
$$
\bz = \bigg[ I + \frac{ \proj{\bw} }{1 - |\bw|^2} \bigg] \by
= \by+\frac{(\bw\cdot\by)\bw}{1-|\bw|^2}.
$$
Inserting this into $\beta = - \bz \cdot \bw$ yields \eqref{prebur}. \qed

\subsubsection{Lemmas for extreme points}
 
\begin{lemma} \label{lemm:s0}
Let $0\le s\le1$, $P = I + \bw \dtsig$ and $Q = I + \bx \dtsig$ with $|\bw|,|\bx|<1$. Let
$\bu = \bw - s \bx$ and $\bv = \bw + s \bx$. Then for every $\by\in\bR^3$,
\begin{align}\label{gen22forma}
&\Tr(\by \dtsig) \frac{1+s}{L_P + s R_Q} (\by \dtsig) \nonumber\\
&\quad= 2(1+s)^2\,\biggl\bra \by, \left[ \bigl\{(1+s)^2-|\bu|^2\bigr\}I
+ | \bu \ket \bra \bu | - | \bv \ket \bra \bv |
- \frac{| \bu \times \bv \ket \bra \bu \times \bv|}
{(1+s)^2 - | \bv|^2} \right] ^{-1} \by \biggr\ket,
\end{align}
where the operator inside $[\quad]^{-1}$ of \eqref{gen22forma} is positive and invertible.
\end{lemma}

\proof
As above, let $\beta I+ \bz \dtsig = \big( L_P + s R_Q \big)^{-1} (\by \dtsig) $ so that 
\begin{align}  \label{eq:s0a}
\by \dtsig & = \big( L_P + s R_Q \big)( \beta I+ \bz \dtsig ) \nn\\
& = [(1+s)\beta + \bz \cdot (\bw + s \bx)] I
+ [(1+s) \bz + \beta (\bw + s \bx) - i \bz \times (\bw - s \bx)] \dtsig \nn\\
& = [(1+s)\beta + \bz \cdot \bv] I + [(1+s) \bz + \beta \bv - i \bz \times \bu] \dtsig.
\end{align}
Since the coefficient of $I$ on the right side of \eqref{eq:s0a} must be $0$, we find
$$
\beta = - \frac{\bz \cdot \bv}{1 + s}.
$$
Inserting this into \eqref{eq:s0a} and equating real and imaginary parts yield
\bsq \begin{eqnarray}
\by & = & \left[ (1+s) I - \frac{ | \bv \ket \bra\bv |}{1+s} \right] \bz_1
+ \bz_2 \times \bu,  \label{eq:s0Ba} \\
0 & = & \left[ (1+s) I - \frac{ | \bv \ket \bra \bv |}{1+s} \right] \bz_2
- \bz_1 \times \bu,  \label{eq:s0Bb}
\end{eqnarray}  \esq 
where we have written $\bz = \bz_1 + i \bz_2$ and $(\bz \cdot \bv ) \bv = \proj{\bv} \bz $.
Solving \eqref{eq:s0Bb} for $\bz_2$ with use of \eqref{eq:inv.3dim} yields
\begin{eqnarray*}
\bz_2 & = & \frac{1}{1+s} \left[ I + \frac{ \proj{\bv} }
{(1+s)^2 - |\bv|^2} \right] (\bz_1 \times \bu).  \\
\end{eqnarray*}
Inserting this into \eqref{eq:s0Ba} gives
\begin{eqnarray}  \label{eq:s0Bc}
\by = \frac{1}{1+s} \left[ \bigl\{(1+s)^2-|\bu|^2\bigr\} I
+ | \bu \ket \bra \bu | - | \bv \ket \bra \bv |
- \frac{| \bu \times \bv \ket \bra \bu \times \bv|}
{(1+s)^2 - | \bv|^2} \right] \bz_1,
\end{eqnarray} 
where we have used
\begin{align*}
&(\bz_1 \times \bu)\times\bu=-(\bu\cdot\bu)\bz_1 +(\bu \cdot \bz_1)\bu, \\
&[|\bv\>\<\bv|(\bz_1\times\bu)]\times\bu=[\bv\cdot(\bz_1\times\bu)]\, \bv\times\bu
= - \proj{\bv \times  \bu }\, \bz_1
\end{align*}
in the first and the second terms from $\bz_2$, respectively. Thus we have proved that for
every $\by\in\bR^3$ there exists a $\bz_1\in\bR^3$ satisfying \eqref{eq:s0Bc}. This implies
that the operator inside $[\quad]$ of \eqref{eq:s0Bc} is surjective and hence invertible.
Since
\begin{eqnarray*}
\tr\, (\by \dtsig) \frac{1}{L_P + s R_Q} (\by \dtsig)  
= \tr\, (\by \dtsig) (\beta I + \bz \dtsig) = 2\by \cdot \bz_1,
\end{eqnarray*}
we obtain \eqref{gen22forma} by solving for $\bz_1$ in \eqref{eq:s0Bc}. Moreover, since the
LHS of \eqref{gen22forma} $\ge0$, the operator inside $[\quad]^{-1}$ is indeed positive. \qed

\medskip
In the case where $\by = \bw - \bx \in \hbox{span} \{ \bu,\bv \}$ and so $\by$ is orthogonal
to $\bu \times \bv$, we can simplify the expression above.

\begin{lemma} \label{lemm2:s0}
In the notation of Lemma~\ref{lemm:s0}, when $\by$ is orthogonal to $\bu \times \bv$ equation 
\eqref{gen22forma} becomes
\begin{eqnarray}\label{gen22formb}
2(1+s)^2\, \Big\bra \by, \big[ \big\{(1+s)^2 - |\bu|^2 \big\} I
+ \proj{\bu} - \proj{\bv} \big] ^{-1} \by \Big\ket.
\end{eqnarray}
\end{lemma}

\proof
First note that when $X \ge 0 $ and $ I - W - X \ge 0 $ is invertible, then $ I - W \ge 0 $
is also invertible. Then it suffices to observe that when $ WX = 0$ and $X\by=0$,
$$
(I - W - X )^{-1} \by = \sum_{k = 0}^\infty (W + X)^k \by = \sum_{k = 0}^\infty W^k \by
=(I-W)^{-1}\by, 
$$
and apply this with
$$
W = \frac{ \proj{\bv} - \proj{\bu} } {(1+s)^2-|\bu|^2},\qquad
X = \frac{| \bu \times \bv \ket \bra \bu \times \bv|}
{\bigl[(1+s)^2-|\bu|^2\bigr]\bigl[(1+s)^2-|\bv|^2\bigr]}. \qed
$$

Note that when $P=Q$ so that $\bu=(1-s)\bw$ and $\bv=(1+s)\bw$, the $s$-dependence of all
terms in \eqref{gen22formb} has the form $(1\pm s)^2/(1+s)^2=(1\pm s^{-1})^2/(1+s^{-1})^2$,
which implies \eqref{eq:Psym} for $A = \by \dtsig$ as
$$
\frac{1+s}{L_P+sR_P}=\frac{1+s^{-1}}{L_P+s^{-1}R_P}=\frac{1+s}{R_P+sL_P}.
$$

    \pagebreak

\end{document}